%% file: main.tex
\documentclass[a4paper,11pt]{article}
\usepackage[margin=1in]{geometry}
\usepackage{graphicx} 
\usepackage{authblk}
\title{More Efforts Towards Fixed-Parameter Approximability of Multiwinner Rules}
\author[1]{Sushmita Gupta}
\author[2]{Pallavi Jain}
\author[1]{Souvik Saha}
\author[1]{Saket Saurabh}
\author[1]{Anannya Upasana}

\affil[1]{The Institute of Mathematical Sciences, HBNI}
\affil[ ]{\texttt{\{sushmitagupta,souviks,saket,anannyaupas\}@imsc.res.in}}
\affil[2]{Indian Institute of Technology Jodhpur}
\affil[ ]{\texttt{pallavij@iitj.ac.in}}

\usepackage{times}
\usepackage{soul}
\usepackage{url}
\usepackage[hidelinks]{hyperref}
\usepackage[utf8]{inputenc}
\usepackage[small]{caption}
\usepackage{graphicx}
\usepackage{amsmath}
\usepackage{amsthm}
\usepackage{booktabs}
\usepackage[authoryear]{natbib}
\usepackage[switch]{lineno}
\input{macro}
\newcommand{\shortcite}[1]{\citeyear{#1}}

\newtheorem{theorem}{Theorem}
\date{}
\begin{document}

\maketitle

\begin{abstract}

\input{abstract-new}

\end{abstract}

\input{intro}

\input{prelims}

\input{algo}

\input{kernel}
\input{additiveapprox}
\input{parabyt}
\input{outlook}

\clearpage
\bibliographystyle{plainnat}
\bibliography{references}

\end{document}

%% file: macro.tex
\usepackage{mathrsfs}
\usepackage{tabularx}

\usepackage{latexsym}
\usepackage{amsmath}
\usepackage{amsfonts}
\usepackage{amssymb}
\usepackage{epsfig}
\usepackage{xfrac}
\usepackage{xspace}
\usepackage{thm-restate}
\usepackage{nicefrac}
\usepackage{subcaption}
\usepackage{enumerate}
\usepackage{enumitem}

\usepackage{amsmath}
\usepackage{amsfonts}
\usepackage{graphicx}
\usepackage{tcolorbox}
\usepackage{amssymb}
\usepackage{multirow}
\usepackage{soul}
\usepackage{colortbl}
\usepackage{algorithm}
\usepackage[noend]{algpseudocode}
\usepackage{amsthm}
\usepackage{mathtools}
\usepackage{mathrsfs}

\usepackage{cleveref}
\usepackage{thmtools} 
\usepackage{thm-restate}

\newtheorem{lemma}{Lemma}

\newtheorem{theorem*}[lemma]{Theorem*}
\newtheorem{corollary}[lemma]{Corollary}
\newtheorem{definition}[lemma]{Definition}

\newtheorem{claim}{Claim}
\newtheorem{proposition}[lemma]{Proposition}
\newtheorem{reduction rule}{Reduction Rule}

\usepackage{comment}
\allowdisplaybreaks

\usepackage{bm}

\newcommand{\hidestuff}[1]{}
\newcommand{\hide}[1]{}

\newcommand{\shortv}[1]{}
\newcommand{\approval}{approval list\xspace}

\newcommand{\Iden}{\ensuremath{{\rm Iden}}}

\newcommand{\cO}{\mathcal{O}}

\newcommand{\cA}{\mathcal{A}}

\newcommand{\cB}{\mathcal{B}}
\newcommand{\cC}{\mathcal{C}}

\newcommand{\cF}{\mathcal{F}}

\newcommand{\Ceil}[1]{\ensuremath{\left\ceil{#1\right}}}

\newcommand{\OO}{{\mathcal O}}

\usepackage{tcolorbox}
\usepackage{tikz}

\usepackage[framemethod=tikz]{mdframed}

\definecolor{mycolor}{rgb}{0.122, 0.435, 0.698}
\newmdenv[innerlinewidth=0.02pt, roundcorner=4pt,linecolor=mycolor,innerleftmargin=6pt,
innerrightmargin=6pt,innertopmargin=6pt,innerbottommargin=6pt]{mybox}
\newmdenv[innerlinewidth=0.5pt, roundcorner=4pt,linecolor=black,innerleftmargin=6pt,
innerrightmargin=6pt,innertopmargin=6pt,innerbottommargin=6pt]{myboxblack}
\newmdenv[innerlinewidth=0.5pt, roundcorner=4pt,linecolor=mycolor,innerleftmargin=6pt,
innerrightmargin=6pt,innertopmargin=6pt,innerbottommargin=6pt]{myboxthick}

\newcommand{\no}{{\sf no}\xspace}
\newcommand{\yes}{{\sf yes}\xspace}
\newcommand{\opt}{{\sf OPT}\xspace}

\newcommand{\fpt}{{\sf FPT}\xspace}

\newcommand{\kddfree}{{$K_{d,d}$-free}\xspace}

\newcommand{\kdd}{{$K_{d,d}$}\xspace}

\newcommand{\whard}{\textsf{W[1]}-hard\xspace}
\newcommand{\wtwohard}{\textsf{W[2]}-hard\xspace}
\newcommand{\nph}{\textsf{NP}-hard\xspace}
\newcommand{\igraph}{profile graph\xspace}
\newcommand{\C}{{{\cal C}}\xspace}
\newcommand{\V}{{{\cal V}}\xspace}
\newcommand{\abc}{{\texttt{ABC}}\xspace}

\newcommand{\defparprob}[4]{
\begin{tcolorbox}[ colback=gray!5!white,colframe=gray!75!black]
  \vspace{-1mm}
  \begin{tabular*}{\textwidth}{@{\extracolsep{\fill}}lr} #1  & {\bf{Parameter:}} #3 \\ \end{tabular*}
  {\bf{Input:}} #2  \\
  {\bf{Question:}} #4
  \vspace{-1mm}
\end{tcolorbox}
}

\newcommand{\owa}{{\sf OWA}\xspace}
\newcommand{\probone}{{\sc SM-MwE}\xspace}
\newcommand{\probonegen}{{\sc SM-MwE}\xspace}

\newcommand{\solution}{{\sf Thiele Committee}\xspace}

\newcommand{\score}{{\sf sco}\xspace}

\newcommand{\probonemax}{{\sc Max SubMod-MwE}\xspace}
\newcommand{\kabfree}{{$K_{a,b}$-free}\xspace}

\newcommand{\floor}[1]{\lfloor #1 \rfloor}
\newcommand{\ceil}[1]{\lceil #1 \rceil}
\newcommand{\appa}{\ensuremath{{\sf APPA}}\xspace}
\newcommand{\optapx}{{\sf ApxOPT}\xspace}
\newcommand{\core}{{\sf Co}\xspace}
\newcommand{\fptas}{{\sf FPT-AS}\xspace}
\newcommand{\oneadapp}{{one-additive approximation}\xspace}
\newcommand{\addapxth}{{\sf Add}-{\sf Apx}-{\sf MwE}\xspace}
\newcommand{\apxvot}{{\sf Apx}-{\sf MwE}\xspace}

\newcommand{\sse}{\subseteq}
\newcommand{\Co}[1]{\ensuremath{\mathcal{#1}}}
\newcommand{\sm}{\setminus\!}
\newcommand{\contribution}[2]{\ensuremath{\lambda_{#1}^{#2}}}

\newtheorem*{lem*}{Lemma}

%% file: abstract-new.tex
{\sc Multiwinner Elections} 
have emerged as a prominent area of research with numerous practical applications. We contribute to this area by designing parameterized approximation algorithms and also resolving an open question by Yang and Wang [AAMAS'18]. More formally, given a set of candidates, \( \mathcal{C} \), a set of voters,  \( \mathcal{V} \), approving a subset of candidates (called approval set of a voter), and an integer $k$,  we consider the problem of 
selecting a ``good'' committee using Thiele rules. 
This problem is computationally challenging for most Thiele rules with monotone submodular satisfaction functions, as there is no \((1-\frac{1}{e}-\epsilon)\)\footnote{Here, $e$ denotes the base of the natural logarithm.}-approximation algorithm in \( f(k)(|\mathcal{C}| + |\mathcal{V}|)^{o(k)} \) time for any fixed \(\epsilon > 0\) and any computable function $f$, and no {\sf PTAS} even when the length of approval set is two. Skowron [WINE'16]  designed an approximation scheme running in FPT time parameterized by the combined parameter, \emph{size of the approval set} and $k$. In this paper, we consider a parameter $d+k$ (no \( d \) voters approve the same set of \( d \) candidates), where $d$ is upper bounded by the size of the approval set (thus, can be much smaller). 
  With respect to this parameter, we design parameterized approximation schemes, a lossy polynomial-time preprocessing method, and show that an extra committee member suffices to achieve the desired score (i.e., $1$-additive approximation). Additionally, we resolve an open question by Yang and Wang~[AAMAS'18] regarding the fixed-parameter tractability of the problem under the PAV rule with the total score as the parameter, demonstrating that it admits an FPT algorithm.

%% file: intro.tex
\section{Introduction}

{\sc Multiwinner Election} is one of the well-studied problems in computational social choice theory~\cite{Ballotpedia,16de4c2a459441008078dd35182e783f,DBLP:journals/scw/ElkindFSS17,DBLP:conf/ijcai/PierczynskiS19}; and the most extensively studied and commonly implemented in practice is the approval-based model of election~\cite{DBLP:conf/aaai/SkowronF15,DBLP:journals/jair/SkowronF17,DBLP:journals/iandc/Skowron17,DBLP:journals/jet/LacknerS21,DBLP:series/sbis/LacknerS23,DBLP:conf/ijcai/DoHL022}.
In general, a multi-winner election with approval preferences consists of a set of $m$ candidates ($\C$), a set of $n$ voters ($\V$), each providing a set of approved candidates $A_v \subseteq \C$,  a satisfaction (or scoring) function $\score:2^{\C} \rightarrow \mathbb{Q}_{\geq 0}$, and an integer $k$. The set of \approval of all voters is called the \emph{approval profile} denoted by $A=\{A_v \colon v\in {\cal V}\}$. 
The goal here is to select a subset (called a {\it committee}) $S$ of $k$ candidates that maximizes the value $\score(S)$. In the decision version the goal is to check if $\score(S)\geq t$ for a given value $t$. The definition of the $\score(\cdot)$ function depends on the voting rule we employ.  In this article, we consider a subclass of approval-based voting rules (known as the \abc voting rules). An important class of \abc rules is the one defined by Thiele~\cite{DBLP:series/sbis/LacknerS23} (also known as \emph{generalised approval procedures}). Some of the well-known Thiele rules are approval voting, Chamberlin-Courant, and proportional approval voting(PAV)[~\cite{RePEc:cup:apsrev:v:77:y:1983:i:03:p:718-733_24,janson2016phragmen}]. 
A Thiele rule is given by a function $f\colon \mathbb{N}\cup \{0\} \rightarrow \mathbb{Q}_{\geq 0}$ where $f(0)=0$. For example, under the Chamberlin-Courant rule, \( f(i) = 1 \) for each \( i > 0 \), while for the Proportional Approval Voting (PAV) rule, \( f(i) = \sum_{j=1}^{i} \frac{1}{j} \) for each \( i > 0 \). Given a profile $A$, the score of a committee $S \subseteq C$ is defined by $\score(S) = \sum_{v\in \V}  f(|S\cap A_v|)$. Since our paper deals with different functions, as well as different approval profiles, for the sake of clarity, we denote the score function as $\score_f(A,S) = \sum_{v\in \V}  f(|S\cap A_v|)$. We leave $f$ and $A$ from the notation if it is clear from the context.

\medskip

\noindent 
{\bf Context of Our Results.} In general, the {\sc  Multiwinner Election} problem, aka  the {\sc Committee Selection} problem, is \nph. In fact, the problem is also intractable in the realm of parameterized complexity, \wtwohard with respect to $k$, the size of the committee. That is, we do not expect \hide{the problem to admit} an algorithm with running time $h(k)(n+m)^{\cO(1)}$. In fact, these intractability results carry over even for special cases. 
In particular, given a Thiele function $f\colon \mathbb{N}\cup \{0\} \rightarrow \mathbb{Q}^{+}$, for each voter $v\in \V$, we can  associate a {\em satisfaction function} with each committee, defined as $f_v:2^{\C} \rightarrow \mathbb{Q}^{+}$ where $f_v(S)=f(|S\cap A_v|)$. In this notation, the score of $S$ is given by $\score_f(A,S) = \sum_{v\in \V}  f_v(S)$. When the function $f_v$ is monotone and submodular, for each voter $v \in V$, we call the problem {\sc Submodular  Multiwinner Election} (\probonegen).
It is known that \probonegen is \nph as well as \wtwohard with respect to parameter $k$ \cite{DBLP:conf/atal/AzizGGMMW15,DBLP:journals/aamas/YangW23}. 
Aziz et el.~\shortcite{DBLP:conf/atal/AzizGGMMW15} shows that the problem under PAV rule remains \wtwohard even for the special case where each voter approves at most two candidates. A similar hardness for the Chamberlin-Courant rule was shown by~\cite{DBLP:journals/aamas/YangW23}. They also show that the problem for both Chamberlin-Courant rule and PAV remains  \whard when parameterized by $(|\cC|-k)=m-k$ even when every voter approves two candidates. They also show that under Chamberlin-Courant rule, the problem admits an \fpt algorithm parameterized by $t$ and under PAV, the problem is \fpt parameterized by the combined parameter $t$ and maximum size of the approval list of a voter. Here, $t$ is the threshold value in the decision version of the problem. However, they state that the \fpt membership of the problem under the PAV rule with respect to $t$ as an open question. We address this by providing an \fpt algorithm for \probonegen, parameterized by $t$ and the committee size. We then show that the same algorithm yields an \fpt algorithm parameterized by $t$ for the problem under the PAV rule. Our approach employs the color-coding technique introduced by \cite{DBLP:journals/jacm/AlonYZ95}, and the detailed algorithm is presented in Section~\ref{sec:parabyt}.

In the world of approximation algorithms, due to~\cite{DBLP:conf/soda/Manurangsi20} we know that we cannot hope to find a $k$-sized committee with score at least $(1-\frac{1}{e}-\epsilon)\opt$ even in time $f(k)(|\C|+|\V|)^{o(k)}$. Here, {\sf OPT} denotes the maximum score of a committee of size $k$. To mitigate these intractability results Skowron~\shortcite{DBLP:journals/iandc/Skowron17} considers special cases of this problem. In particular, he looks at the case where the approval list is bounded by an integer $\delta$, that is, every voter approves at most $\delta$ candidates. It was already known that even for $\delta=2$, \probonegen is \nph and \whard~\cite{DBLP:conf/atal/AzizGGMMW15,DBLP:journals/jair/BetzlerSU13}. This led Skowron to consider the existence of {\it parameterized approximation} that is, an approximation algorithm that runs in time $h(k,\delta) (n+m)^{\cO(1)}$ where $h$ is any computable function. In particular, he observed that the problem admits a $(1-\frac{1}{e})$-factor approximation in polynomial time. In addition, for each $\epsilon >0$, he presented, an approximation scheme that runs in time $h(k,d,\epsilon) (n+m)^{\cO(1)}$ and produces a $k$-sized subset $S$ such that $\score_f(A,S)\geq (1-\epsilon){\sf OPT}$. We call such an approximation scheme an \fptas and is the starting point of this work.

It is known that even when $\Delta_C =3$ (maximum number of voters approving the same candidate) and $\delta =2$ (maximum number of candidates approved  by the same voter) \probonegen is \nph \cite{DBLP:conf/ijcai/ProcacciaRZ07,DBLP:conf/atal/AzizGGMMW15}. We show that \probonegen is \fpt when parameterised by $k+\Delta_C$ (Section \ref{sec:parabyt} ).

\medskip

\noindent 
{\bf Our results and overview.} For our study of parameterized approximation algorithms, we consider a parameter $d$ smaller than $\delta$ as well as $\Delta_C$. Here, $d$ denotes the smallest number such that that no $d$ voters approve the same set of $d$ candidates. Clearly, $d\leq \delta$ as well as $d\leq \Delta_C$. Since it is a smaller parameter, it is worth considering. There are realistic scenarios where $d$ is indeed much smaller than $\delta$ and $\Delta_C$. Consider a university election of a 10-member committee from 200 candidates, with 5,000 students voting, where the votes are presumably based on personal connections and shared interests. The large diverse student body makes it unlikely for any \( d \) students to approve the same \( d \) candidates, thereby resulting in unique voting patterns.

For the ease of exposition, we consider the \igraph of approval-based elections. It is a bipartite graph $G=(\C,\V,E)$, where $V(G)=\C \uplus \V$, and $E$ is the edge set. Note that  $\C$ is the set of candidates, $\V$ is the set of voters. For a candidate $c\in \C$ and a voter $v\in \V$, we add an edge $cv$ if the voter $v$ approves the candidate $c$, i.e $c\in A_v$. 
In \igraph the approval list of voter $v$,$A_v$ is the set $\{c\in \C \colon vc \in E(G) \}$. Given a graph $G$ and a function $f_v$ for every $v\in \V$, the objective function is same as above, i.e., find a subset $S\subseteq \C$ of size $k$  that maximises $\score_f(G,S)=\sum_{v\in \V}f_v(|S\cap A_v|)$. The graph $G$ is $K_{d,d}$-free (i.e., it does not contain a complete bipartite graph with $d$ vertices on each side as an induced subgraph).

Jain et al.~\shortcite{DBLP:conf/soda/0001KPSS0U23} consider the {\sc Maximum Coverage} problem (which is equivalent to Chamberlin-Courant based {\probonegen}) and give an \fptas with respect to the parameter $k+d$ when the \igraph is $K_{d,d}$-free. Manurangasi~\shortcite{DBLP:journals/tcs/Manurangsi25} designed a polynomial-time \emph{lossy kernel} for the same problem.

Similar to the results of Jain et al.~\shortcite{DBLP:conf/soda/0001KPSS0U23} for {\sc Maximum Coverage} we obtain the following set of results for \probonegen, when \igraph of the given instance is \kddfree. In the following a {\it solution
} refers to a $k$-sized committee.
  \begin{itemize}

 \item We present an \fptas parameterized by $k$. That is we give an algorithm that given $0 < \epsilon < 1$, runs in time $(\frac{dk}{\epsilon})^{\cO(d^2k)} (n+m)^{\cO(1)}$, and outputs a solution \hide{$k$-sized committee} whose score is at least $ (1-\epsilon)$ fraction of the optimum, \Cref{thm:fpt-apx}.

\item We complement \fptas, by designing a polynomial time {\it lossy kernel} with respect to $k$ for \probonegen. (Observe that a normal ``lossless'' kernel is not possible with respect to $k$, since the problem is \whard, and thus no \fpt algorithm and equivalently a kernel, may exist.) 
In other words, we present a polynomial-time algorithm that produces a graph $G'$ of size polynomial in $k+\epsilon$ from which we can find a solution that attains a $(1 - \epsilon)$ fraction of the \hide{maximum score of any $k$-sized committee} optimal score in the original instance, \Cref{thm:lossy-kernel}. Observe that in most practical scenarios, $k$ is some fixed small constant and thus searching for the desired committee in the reduced instance is quite efficient. Moreover, we also note that $G'$ represents the reality that only a small subset of voters and candidates actually matter! 

The starting point of our lossy kernel is the result of \cite{DBLP:journals/corr/abs-2403-06335}. In particular, the kernelization algorithm involves the following steps: it assesses the potential value of the score, and if this value is upper-bounded by a polynomial function of $k$ and $\epsilon$, we can then construct a lossless kernel. Else, we first define a notion that leads to a reduction rule for identifying candidates who are similar in terms of approximating the optimal score. Exhaustive application allows us to reduce the size of the candidate set to a polynomial function of $k$ and $\epsilon$. Finally, by applying another reduction rule that identifies distinct approval lists, we can reduce the number of voters, resulting in the desired lossy kernel.

\item We also present an \fpt approximation algorithm parameterized by $k,d$ that outputs a $k+1$-sized committee whose score is the same as the optimal solution of size $k$, \Cref{theorem:addapprox}.

\end{itemize}

\noindent In fact, our algorithm works even when we have different Thiele functions $h_v$ for each voter $v$. Then the associated {\em satisfaction function} $f_v:2^{\C} \rightarrow \mathbb{Q}^{+}$ can be represented as $f_v(S)=h_v(|S\cap A_v|)$. This strictly generalizes the known model as for each $v$, $h_v$ is the same.  Our results build on existing algorithms for \textsc{Max Coverage} and \textsc{Maximizing Submodular Functions}. However, due to the inherent generality of our problems, both in terms of the scoring function and the class of profiles, we must deviate significantly from known approaches in several crucial and key aspects.

%% file: prelims.tex
\section{Our problem in \owa framework}\label{section:owaframework}

We will reformulate our problem within the Ordered Weighted Average (\owa) framework to align with existing results in the literature. We refer to Section~$4.1$ in the article of Skowron~\cite{DBLP:journals/iandc/Skowron17} for further details.  Given a set of candidates $\C$, a set of voters $\V$ and the \approval $A_{v}$ of every voter $v\in \V$, and a Thiele rule given by a non-decreasing function $f\colon \mathbb{N}\cup \{0\} \rightarrow \mathbb{Q}^{+}$ with $f(0)=0$, we define an \owa vector~$\lambda$ as follows. For every $i\geq 1$,
 $\lambda_i=f(i)-f(i-1)$. Then, we have $f(|S\cap A_{v}|)=\sum_{j=1}^{|S\cap A_{v}|}\lambda_{j}$. Thus, the scoring function can also be expressed as $\score_G(S)=\sum_{v \in \V}\sum_{j=1}^{|S\cap A_{v}|}\lambda_{j}$. Hence, every Thiele rule can be expressed by an \owa vector. For example the CC rule is represented by $\{1,0,\dots,0\}$ and PAV is represented by $\{1,\frac{1}{2},\dots,\frac{1}{k}\}$. Now, we make the following claim. 
 \begin{lemma}
\label{lem:submodular-non-increasing OWA equivalence}
    The functions $f_v$, $v\in \V$, is monotone and submodular if and only if the corresponding \owa vector $\lambda$ is non-increasing. Here, $f_v$ is the satisfaction function corresponding to the voter $v$ derived from $f$. 
 \end{lemma} 
\begin{proof}
    Consider a voter $v \in \V$ such that its satisfaction function $f_v$ is submodular. Suppose that the corresponding  \owa vector has entries such that $\lambda_i < \lambda_{i+1} $ for some $i\geq 1$. 
     Let $S$ denote a committee formed by taking exactly $i-1$ candidates from $A_v$. Let candidates $x_1,x_2 \in A_v$ such that $x_1,x_2 \notin S$. Thus, it follows that
\begin{align*}
f_v(S \cup \{x_1\}) + f_v(S \cup \{x_2\}) &= \sum_{i\in [2]}f(|S+x_i|\cap A_v)\\
&\hspace{-1.5cm} = 2\cdot f(i)= 2\cdot (f(i-1)+\lambda_{i})\\
& \hspace{-1.5cm}< f(i-1) + \lambda_{i} + \lambda_{i+1} + f(i-1)\\
& \hspace{-1.5cm} = f_v(S \cup \{x_1,x_2\})+f_v(S)
             \end{align*}
But this contradicts the submodularity of $f_v$. Hence, we can conclude that the \owa vector $\lambda$ is non-increasing.

Next, we will prove the other direction. For an arbitrary voter $v\in \V$, let $S$ denote a committee such that $|S\cap A_v|=i$ for some $i\in [|A_v|]$. 

Consider a subset $S \sse \C$ and a pair of distinct candidates $x_1, x_2 \in \C \sm S$. We will show that $f_v(S\cup \{x_1\}) + f_v(S\cup \{x_2\}) \geq f_v(S\cup \{x_1, x_2\}) + f_v(S)$. To argue this we note that $|S\cap A_v|=i$, for some $i\in [|A_v|]$.

\begin{align*}
\sum_{j\in [2]}f_v(S\cup \{x_j\})  &= 2 \cdot f(i+1)\\
& [\text{ because $x_1, x_2 \notin S$ and $x_1\neq x_2$.}]\\
& = 2\cdot (f(i)+ \lambda_{i+1})\\
&\geq 2\cdot f(i)+\lambda_{i+1} +\lambda_{i+2}\\
& =f_v(S \cup \{x_1,x_2\})+f_v(S)
\end{align*}
\end{proof}

\shortv{
  \begin{proof}Consider a voter $v \in \V$ such that its satisfaction function $f_v$ is submodular. Suppose that the corresponding  \owa vector has entries such that $\lambda_i < \lambda_{i+1} $ for some $i\geq 1$. 
     Let $S$ denote a committee formed by taking exactly $i-1$ candidates from $A_v$. Let candidates $x_1,x_2 \in A_v$ such that $x_1,x_2 \notin S$. Thus, it follows that
\begin{align*}
f_v(S \cup \{x_1\}) + f_v(S \cup \{x_2\}) &= \sum_{i\in [2]}f(|S+x_i|\cap A_v)\\
&\hspace{-1.5cm} = 2\cdot f(i)= 2\cdot (f(i-1)+\lambda_{i})\\
& \hspace{-1.5cm}< f(i-1) + \lambda_{i} + \lambda_{i+1} + f(i-1)\\
& \hspace{-1.5cm} = f_v(S \cup \{x_1,x_2\})+f_v(S)
             \end{align*}
But this contradicts the submodularity of $f_v$. Hence, we can conclude that the \owa vector $\lambda$ is non-increasing.

Next, we will prove the corollary. For an arbitrary voter $v\in \V$, let $S$ denote a committee such that $|S\cap A_v|=i$ for some $i\in [|A_v|]$. 

Consider a subset $S \sse \C$ and a pair of distinct candidates $x_1, x_2 \in \C \sm S$. We will show that $f_v(S\cup \{x_1\}) + u_v(S\cup \{x_2\}) \geq f_v(S\cup \{x_1, x_2\}) + f_v(S)$. To argue this we note that $|S\cap A_v|=i$, for some $i\in [|A_v|]$.

\begin{align*}
\sum_{j\in [2]}f_v(S\cup \{x_j\})  &= 2 \cdot f(i+1)\\
& [\text{ because $x_1, x_2 \notin S$ and $x_1\neq x_2$.}]\\
& = 2\cdot (f(i)+ \lambda_{i+1})\\
&\geq 2\cdot f(i)+\lambda_{i+1} +\lambda_{i+2}\\
& =f_v(S \cup \{x_1,x_2\})+f_v(S)
\end{align*}
\end{proof}
}

\noindent We design our algorithms for the scenario where each voter has their own Thiele function\footnote{Note that we are presenting algorithms for the general model, however, this is the first work even when the Thiele function is the same for every voter.}. Thus, instead of a single \owa\ vector, we work with a family  $\Lambda=\{\lambda^v ~|~v\in \V\}$ of \owa\ vectors. Throughout this paper, we assume that for any given non-increasing  \owa vector $\lambda^v=(\lambda_1^v,\lambda_2^v,\ldots, \lambda^v_{|A_v|} )$ we have $\lambda_1^v\leq 1$ for any $v \in \V$. Suppose not, then we can divide every $\lambda_i^v$ by $\lambda_{\max}=\max\{\lambda_1^v\mid v \in \V \}$ and change $t$ to $\frac{t}{\lambda_{\max}}$ where $t$ corresponds 
 to total score and is described in problem definition below.  Given a bipartite graph $G=(\C,\V,E)$ and a set $\Lambda$ of \owa vectors, $\lambda^v$, for every $v \in \V $, we define  a restriction  operation $\Lambda_S$ for every $S\subseteq \C$.

\begin{definition}
For any $j\in [|A_v|]$, we use $\lambda^v_{-j}$ to denote the (shortened) \owa vector obtained by deleting the $j$-sized prefix of $\lambda^v$. That is, $\lambda^v_{-j} =(\lambda^v_{j+1},\lambda^v_{j+2},\ldots, \lambda^v_{|A_v|})$.

For the set of \owa vectors $\Lambda$ and a subset of candidates $S\sse \Co{C}$, we define
$\Lambda_S = \{ \lambda^v_{-j}~: v \in \V \text{ and } ~j = |N(v) \footnote{ $N(\cdot)$ denotes the neighborhood } \cap S|\}$.\hide{Thus, it is the set of \owa vectors obtained from $\Lambda$ by removing the $|N(v)\cap S|$-sized prefix from $\lambda_v$, for each voter $v\in \Co{V}$.} In other words, $\Lambda_S$ is obtained by removing the $|N(v)\cap S|$-sized prefix from $\lambda^v$, if $v\in N(S)$; else we retain $\lambda^v$ intact. 

\end{definition}

Let $\lambda_{\min}$ denote $ \min\{\lambda^v_1 \mid v\in \V \}$
When $\Lambda, G$ are clear from the context, we will use $\score(S)$ or $\score_G(S)$ instead of $\score_G^\Lambda(S)$. For a set $O$ and a singletone set $\{x\}$ we sometimes omit the braces during set operations. For example, $O\backslash \{x\}$ and $O\backslash x$ represent the same set.

%% file: algo.tex
\section{\fptas for \probonegen}

For clarity, we state the problem here.

\defparprob{\probonegen}{A bipartite graph $G=(\C,\V,E)$, a set $\Lambda$ of  non-increasing vectors $\lambda^v=(\lambda_1^v,\lambda_2^v,\ldots, \lambda_{|A_v|}^v)$ for all $v\in \V$, and positive integers $k$ and $t$.}{$k$}{Does there exist $S\subseteq \C$ such that $|S|\leq k $ and $\score_G^\Lambda(S)=\sum_{v \in \V}f_{G,v}(S)\geq t$ where $f_{G,v}(S)=\sum_{j=1}^{|N_G(v)\cap S|}\lambda_j^v$ ? }

\paragraph{Overview of the algorithm.} We derive our algorithm by considering two cases: "low" threshold (value of $t$) and "high" threshold. For "low" threshold, we use a sunflower lemma-based reduction rule to reduce candidates. For "high" threshold, we discard all but a sufficiently large number of candidates with the highest $\score(\cdot)$ value. The formal description is presented in~\Cref{alg:combgenapxthn}.

\begin{algorithm}[t]
\caption{\apxvot : An \fpt-approximation scheme for \probonegen}\label{alg:combgenapxthn}
\textbf{Input:} A bipartite graph $G=(\C,\V,E)$, a set $\Lambda=\{\lambda^v: v\in \V\}$, integers $k,t$, and $\epsilon>0$   \\
\textbf{Output:} A $k$-sized subset $S\subseteq \C$ such that \hide{$|S|\leq k $ and} $\score_G(S)\geq (1-\epsilon)t$.
\begin{algorithmic}[1]
\Statex Let $r=\frac{4dk}{\epsilon\lambda_{\min}}+k$.
 \If{$t\leq\frac{2kr^d(d-1)}{(r-k)\epsilon}$}
 \State Apply Reduction Rule~\ref{redrule:sunflower} exhaustively to get $G'=(\C',\V)$.
 \State Search all $k$-sized subsets of $\C'$. Let $S\subseteq \C'$ that achieves the maximum \score.
  \If {$\score_{G}(S)\geq t$} \Return $S$.
  \Else { \Return \no-instance}.
  \EndIf
\EndIf
 \If{$t>\frac{2kr^d(d-1)}{(r-k)\epsilon}$}
  \State  Let $C_r$ denote the $\ceil{r}$ vertices of $\C$ with the highest $\score(\cdot)$-value.
  \State Let $S$ denote the $k$-sized subset of $\C_r$ with maximum $\score_{G}(\cdot)$ value.  \label{line:search}
    \If {$\score_{G}(S)\geq (1-\epsilon)t$} \Return $S$.
    \Else { \Return \no-instance}.
    \EndIf
\EndIf
\end{algorithmic}
\end{algorithm}

Next, we will define a sunflower that is used in our proof. In a bipartite graph $G(\C,\V, E)$, a subset $S\subseteq \C$ is said to form a {\it sunflower} if $N_G(x)\cap N_G(x')$ (i.e the ``approving set") are the same for all distinct candidates $x,x' \in S$. For a given sunflower $S$, we will refer to the common intersection of the neighborhood, $\cap_{x\in S} N(x)$ as $\core(S)$.

\begin{proposition}{\rm \cite{DBLP:journals/corr/abs-2403-06335}}\label{lemma:kabfreesunflower}{\rm [$K_{a,b}$-{\bf free sunflower}]}
    For any $w,l\in \mathbb{N}$, let $G((\C,\V),E)$ be a \kabfree bipartite graph such that every vertex in $\C$ has degree at most $\ell$ and $|\C|\geq a((w-1)\ell)^b$. Then, $G$ has a sunflower of size $w$. Moreover, a sunflower can be found in polynomial time.
\end{proposition}

\begin{theorem}\label{thm:fpt-apx}
    There exists an algorithm running in time $(\frac{dk}{\epsilon})^{\mathcal{O}(d^2k)}n^{\mathcal{O}(1)}$ that given an instance of \probone, where the input graph is \kdd-free, outputs a solution $S$ such that $\score_G^{\Lambda} (S)\geq (1-\epsilon)t$.
\end{theorem}

\begin{proof}
We will design and run two different algorithms for two possible cases, based on the value of $t$ (we call it {\it threshold}). \hide{and return the best output. We analyze each case separately and show that the best of the two yields the desired bounds, thereby establishing the correctness of the algorithm.} A brief description of \fptas that combines both cases is given in~\Cref{alg:combgenapxthn}. Recall that we defined $\lambda_{\min}=\min\{\lambda^v_1 \mid v\in \V \}$. Let $r=\frac{4dk}{\epsilon\lambda_{\min}}+k$. 

\begin{description}
    \item[Case 1:]\label{case1} The threshold $t\leq\frac{2kr^d(d-1)}{(r-k)\epsilon}$ 
    \item[Case 2:]\label{case2} The threshold $t >\frac{2kr^d(d-1)}{(r-k)\epsilon}$
\end{description}

\paragraph{Analysis of Case 1.} In this case, we will apply a modified sunflower lemma to reduce the number of candidates and then find an optimal solution using exhaustive search. Thus, for this case, we solve the problem optimally. 

We begin by observing that if there exists a vertex $v \in \C$ with $deg(v)\geq \frac{2kr^{d}(d-1)}{(r-k)\epsilon\lambda_{min}}$, then $\{v\}$ itself is a solution since \hide{(in Case 1)}$t\leq\frac{2kr^d(d-1)}{(r-k)\epsilon\lambda_{\min}}\lambda_{\min}$. This is because each of the neighbors of $v$ contribute at least $\lambda_{\min}$ to $\score_{G}(\{v\})$. Hence, $\score_{G}(\{v\}) \geq \frac{2kr^d(d-1)}{(r-k)\epsilon\lambda_{\min}}\lambda_{\min} \geq t$.
Thus, we may assume that for each $v\in \C$, $deg(v)< \frac{2kr^d(d-1)}{(r-k)\epsilon\lambda_{\min}}$.

Let $W=\frac{2kr^d(d-1)}{(r-k)\epsilon\lambda_{min}}$, the maximum degree of a vertex in $\C$. We apply the following reduction rule to the instance $\Co{I}=(G=(\C,\V,E),k,t,\Lambda)$ exhaustively.

\begin{reduction rule}\label{redrule:sunflower}
Use~\Cref{lemma:kabfreesunflower} on $G$, where $a=b=d$ and $\ell=W$. If a sunflower of size at least $w=Wk+1$ is found, then delete the vertex (candidate), say $u$, with the lowest $\score_G^\Lambda(\{u\})$ value in the sunflower (ties are broken arbitrarily). Return instance $\Co{I'}=(G'=(\C\backslash \{u\},\V),k,t,\Lambda)$.
\end{reduction rule}

 \noindent Lemma~\ref{lemma:optsame} proves the correctness of our reduction rule.
\begin{lemma}\label{lemma:optsame}
    $\Co{I}$  is a \yes-instance iff $\Co{I'}$ is a \yes-instance.
\end{lemma}
    \begin{proof}
If $\Co{I'}$ is a \yes-instance, then clearly $\Co{I}$ is a \yes-instance as well because for any solution $S$ in \Co{I'} we have $\score_{G}(S)= \score_{G'}(S)$. 

Next, for the other direction suppose that $\Co{I}=(G,k,t,\Lambda)$ is a \yes-instance. Let $S$ denote a solution to \Co{I}. Let $u$ denote the vertex in the sunflower $T$ that is deleted by the reduction rule. If $u \in \C\sm S$, then $S$ is a solution in \Co{I'}. 
 
 Suppose that $u\in S$. We will show that there is another candidate that can replace $u$ to yield a solution in \Co{I'}. Formally, we argue as follows. Since $|S|\leq k$, we have $|N(S)|\leq Wk$. For each candidate $x\in T$, we call the set $N(x) \sm \core(T)$ the {\it petal} of $x$. Since $|T|= Wk+1$, there are $Wk+1$ petals in \Co{V}. The voters in $N(S)$ can be present in at most $Wk$ petals. Thus, there is at least one petal, corresponding to some candidate $v\in T$,  that does not contain any vertex in $N(S)$. That is, $(N(v) \sm \core(T)) \cap N(S)= \emptyset$. Using this candidate $v$, we define the set $S'=  S\cup \{v\} \sm \{u\}$. We will next prove that $\score_{G'}(S') \geq \score_G(S)$. Consequently, it follows that $S'$ is a solution in \Co{I'}.

We begin the argument by noting that the petal of $v$, $N(v)\sm\core(T)\sse N(T)\sm N(S)$ and the contribution of the voters in the petal to $S'$ is $\sum_{x\in N(v)\sm \core(T)}\contribution{1}{x}$. The score of $S'$ consists of the score given by voters in $\V \sm (N(u) \cup N(v))$ whose contribution is unchanged between $S$ and $S'$, as are the contributions of the voters in $\core(T)$. The voters who experience a change are in $N(u)\sm \core(T)$ who have one fewer representative in $S'$, and those in $N(v)\sm \core(T)$, who contribute $\sum_{x\in N(v)\sm \core(T)}\contribution{1}{x} $. Following claim completes the proof.

\begin{claim}\label{claim:comparing-S'-to-S}
We show that 
$\score_{G'}(S') \geq \score_{G}(S)$
\end{claim}
\begin{proof}
We note that 
\begin{align*}
\score_{G'}(S')= & \sum_{x\in \V \sm (N(u) \cup N(v))}\sum_{j=1}^{|N(x)\cap S|}\contribution{j}{x} \\
&+ \sum_{x\in \core(T)}\sum_{j=1}^{|N(x)\cap S|}\contribution{j}{x} \\
& +  \sum_{x\in N(u)\sm \core(T)}\sum_{j=1}^{|N(x)\cap S|-1}\contribution{j}{x}+ \sum_{x\in N(v)\sm \core(T)}\contribution{1}{x}
\end{align*}

Next, we see that 
\begin{align*}
\score_{G}(S)= &\sum_{x\in \V \sm (N(u) \cup N(v))}\sum_{j=1}^{|N(x)\cap S|}\contribution{j}{x} +  \sum_{x\in \core(T)}\sum_{j=1}^{|N(x)\cap S|}\contribution{j}{x} \\
&  +  \sum_{x\in N(u)\sm \core(T)}\sum_{j=1}^{|N(x)\cap S|-1}\contribution{j}{x}\\
&+ \sum_{x\in N(u)\sm \core(T)} \contribution{|N(x)\cap S|}{x}
\end{align*}
By definition, $\score_{G}(\{v\}) \geq \score_{G}(\{u\})$. Hence, \[\sum_{x\in N(v)\sm \core(T)}\contribution{1}{x}
\geq \sum_{x\in N(u)\sm \core(T)}\contribution{1}{x} \geq \sum_{x\in N(u)\sm \core(T)} \contribution{|N(x)\cap S|}{x}.\] Therefore, $\score_{G'}(S') \geq \score_{G}(S)$.
\end{proof}

Hence, the lemma is proved.
\end{proof}

Exhaustive application of \Cref{redrule:sunflower} yields an instance \Co{I'} in which a sunflower of size $Wk+1$ does not exist. Then, according to \Cref{lemma:kabfreesunflower}, $|\C'| < d(W^2k)^{d}$ where $a=b=d$, $\ell=W$ and $w=Wk+1$.

\begin{claim} \label{claim:cchoosek}
 The quantity ${|\Co{C'}|\choose k} \leq (\frac{dk}{\epsilon})^{\Co{O}(d^2k)}$.   
\end{claim}
\begin{proof}
By sunflower lemma, we have $|\Co{C'}|\leq d((w-1)\ell)^d$ where $\ell$ is the bound on degree and $w$ is the size of the sunflower. We have $\ell=W=\frac{2kr^d(d-1)}{(r-k)\epsilon\lambda_{min}}$. Here $r=\frac{4kd}{\epsilon\lambda_{min}}+k$. Also the sunflower size $w=Wk+1$. Puting the values of $w,\ell$ in the equation $|\Co{C'}|\leq d((w-1)\ell)^d$ we get 
\begin{align*}
|\Co{C'}|&\leq  d((w-1)\ell)^d\\
&=d(Wk\ell)^d \hspace{20pt} [\text{because } w=Wk+1]\\
&=d(W^2k)^d \hspace{20pt} [\text{because }l=W]
\end{align*}
Now we separately evaluate $W$ first. The second equality is obtained by substituting the value of $r$.
\begin{align*}
    W&=\frac{2kr^d(d-1)}{(r-k)\epsilon\lambda_{min}}
    \hide{=\frac{2kr^d(d-1)}{(\frac{4dk}{\epsilon\lambda_{min}}+k-k)\epsilon\lambda_{min}}}
    =\frac{2kr^d(d-1)}{4dk}\\
    &=\frac{r^d(d-1)}{2d}
    =\frac{(\frac{4dk}{\epsilon\lambda_{min}}+k)^d(d-1)}{2d}
    =(\frac{dk}{\epsilon})^{\mathcal{O}(d)}
\end{align*}
Thus we have $|\Co{C'}| = (\frac{dk}{\epsilon})^{\mathcal{O}(d^2)} $ which implies $\binom{|\Co{C'}|}{k}=(\frac{dk}{\epsilon})^{\mathcal{O}(kd^2)}$
\end{proof}

Notice that \Cref{redrule:sunflower} can be implemented in time polynomial in $|\Co{I}|$, and the number of times it can be applied is also polynomial in $|\Co{I}|$. Hence, the running time in this case is   ${|\Co{C'}|\choose k}n^{\mathcal{O}(1)} = \left(\frac{dk}{\epsilon}\right)^{\Co{O}(d^2k)}n^{\mathcal{O}(1)}$.
This completes the analysis for Case 1. Next, we will analyze Case 2.

\paragraph{Analysis of Case 2.} 

We prove the existence of an approximate solution by showing that starting from an optimal solution, we can create our solution $O_\ell$ in a step-by-step fashion. \Cref{claim:common-neighborhood} allows us to bound for the $i^{th}$ step, the number of voters a top $r$-candidate may share with the candidates in $O_i$. This in turn implies that we can replace candidate $x_{i}$ by someone in $\Co{C}_r\sm O_i$ without too much loss in score. \Cref{claim:neighborhood-size} allows us to give a counting argument that yields that the difference $\score(O)-\score(O_\ell)=\sum_{i=1}^{\ell}\alpha_i \leq \epsilon\cdot \score(O)$. Let $O$ denote a solution for an instance \Co{I} of \probonegen, i.e., $\score(O)\geq t$. \hide{Let $\score(O)=\opt$.} Recall that $C_r$ is defined to be the set of $\ceil{r}$ candidates in $\C$ with the highest \score($\cdot$) value. If $O\sse C_r$, then the exhaustive search of Line~\ref{line:search} will yield the solution $O$. Therefore, without loss of generality, we may assume that $O \sm C_r \neq \emptyset$. 
Let $ O\backslash C_r=\{x_1,\ldots,x_{\ell}\}$, where $\ell \in [k]$.  We define $O_1=(O\backslash \{x_1\}) \cup \{y_1\}$ where $y_1 \in C_r \backslash O$ such that\hide{$\score(O_1)$  is maximum, that is,} $y_1$ minimizes the value $\score(O)-\score(O_1)$. Similarly, for any $i \in [\ell]$, we define $O_i=(O_{i-1} \sm \{x_i\}) \cup \{y_i\}$ where $y_i \in C_r \sm O_{i-1}$ such that $\score(O_{i-1})-\score(O_i)$ is minimum. For each $i\in [\ell-1]$, we define $\alpha_i=\score(O_{i})-\score(O_{i+1})$.

\begin{claim}\label{claim:common-neighborhood}Let $p$ be any candidate in $\C_{r} \sm O_i$. Then, for any $i\in [\ell]$, we have $|N(O_i) \cap N(p)| \geq \alpha_{i-1}$. 
\end{claim}

\begin{proof}
    From the definition, it follows that $\score^\Lambda(O_{i-1}\sm  \{x_i\} \cup \{p\})=\score^\Lambda(O_{i-1} \sm \{x_i\})+\score^{\Lambda_{O_{i-1}\sm \{x_i\}}}(p)$. The term $\score^{\Lambda_{O_{i-1}\sm \{x_i\}}}(p)$ captures the marginal contribution of the candidate $p$ when added to the set $O_{i-1} \sm \{x_i\}$, i.e., the marginal contribution of $p$ to $\score(O_{i-1} \sm \{x_i\} \cup \{p\})$.

Let $\score^{\Lambda_{O_{i-1}\backslash x_i}}(p)=\score_G^\Lambda(p)-Z$.

Another way of accounting for the marginal contribution of $p$ to $O_{i-1} \sm \{x_{i}\}$ is as follows. We note that $\score(O_{i-1} \sm \{x_i\} \cup \{p\})-\score(O_{i-1}\sm  \{x_i\})$ can be expressed as 
\begin{align*}
    &\score(O_{i-1}\sm  \{x_i\} \cup \{p\})-\score(O_{i-1}\sm \{x_i\})\\
    &= \sum_{v\in N(p)\cap N(O_{i-1}\backslash \{x_i\})}\lambda^v_{1+|N(v)\cap N(O_{i-1}\backslash \{x_i\})|} \\
    &+\sum_{v\in N(p)\sm N(O_{i-1}\sm \{x_i\}) }\lambda^v_1. 
\end{align*}

Thus, by equating the two expressions for the marginal contribution of $p$ to $\score(O_{i-1}\sm \{x_i\} \cup \{p\})$, we get 
\[\score(\{p\})-Z = \score(O_{i-1}\sm  \{x_i\} \cup \{p\})-\score(O_{i-1}\sm \{x_i\})\]

On further simplification we can bound 
\begin{align*}
 Z & =  \score(\{p\}) \\
 &- \left( \score(O_{i-1}\sm  \{x_i\} \cup \{p\}) -\score(O_{i-1}\sm \{x_i\}) \right) \\
 & =  \sum_{v\in N(p)} \lambda^v_1 -\\
& \left( \sum_{v\in N(p)\cap N(O_{i-1} \sm \{x_i\})}\lambda^v_{|N(p)\cap (O_{i-1}\sm \{x_{i}\})|+1} \right. \\
& \left.+ \sum_{v\in N(p)\sm N(O_{i-1}\sm \{x_i\}) }\lambda^v_1 \right) \\
 & =   \sum_{v\in N(p)\cap N(O_{i-1} \sm \{x_i\})} \left( \lambda^v_1 -  \lambda^v_{|N(p)\cap (O_{i-1}\sm \{x_{i}\})| +1}  \right) \\
& \leq  |  N(p) \cap N(O_{i-1} \sm \{x_i\}) |
\end{align*}
The last inequality is due to the fact that $\lambda^{v}_i\leq 1$ for each $v\in \V$ and $i\in [|A(v)|]$. 

Note that $N(O_{i-1} \sm \{x_i\}) \cap N(p)  \sse  N(O_{i-1}\sm\{x_i\} \cup \{y_i\})\cap N(p)$. Hence,  $Z\leq  | N(O_{i-1} \sm \{x_i\} \cup \{y_i\}) \cap N(p)|$.  Thus, it is sufficient to show that $\alpha_{i-1} \leq Z$.

Towards this, we begin by noting that 
\begin{align*}
      &\score(O_{i-1}\backslash \{x_i\} \cup \{p\})\\
       &= \score(O_{i-1}\backslash \{x_i\}) +(\score(\{p\})-Z)\\
      & = \score(O_{i-1}) - \\
      &\sum_{v \in N(x_i) \cap N(O_{i-1})} \lambda^v_{|O_{i-1}| } +(\score(\{p\})-Z)\\
      & \geq \score(O_{i-1}) - \sum_{v\in N(\{x_i\})} \lambda^v_{1}  + (\score(\{p\})-Z)\\
      &\geq \score(O_{i-1})-\score(\{x_i\})+\score(\{p\})-Z\\
\end{align*}

By rearranging we have,
\begin{align} 
\begin{split}
&Z -(\score(\{p\}) - \score(\{x_{i}\})) \\
&\geq \score(O_{i-1}) - \score(O_{i-1}\backslash \{x_i\} \cup \{p\}) \label{eq1}
\end{split}
\end{align}

By definition of $y_i$, we know that 
\begin{align}
\begin{split}
&\score(O_{i-1}) - \score(O_{i-1}\backslash \{x_i\} \cup \{p\})\\
& \geq \score(O_{i-1}) - \score(O_{i-1}\sm \{x_i\} \cup \{y_i\})  \\
& = \score(O_{i-1}) - \score(O_{i}) = \alpha_{i-1}  \label{eq2}
\end{split}
\end{align}

By combining \Cref{eq1,eq2} we get $Z - (\score(\{p\}) - \score(\{x_{i}\}) ) \geq \alpha_{i-1}$. We know that $\score(\{p\}) \geq \score(\{x_{i}\})$ since $p\in C_{r}$ and $x_{i}\in O\sm C_{r}$. Thus, it follows that $Z \geq \alpha_{i-1}$.
\end{proof}

\begin{claim} \label{claim:neighborhood-size}
$|N(O\cup O_{\ell})| \leq \frac{2 \cdot \score_{G}(O)}{\lambda_{\min}}$
\end{claim}
\begin{proof}
    Suppose that $|N(O)|>\frac{\score_{G}(O)}{\lambda_{\min}}$. Then, it follows that $\score_G(O)> \frac{\score(O)}{\lambda_{\min}} \lambda_{\min}=\score(O)$, a contradiction. A similar argument yields $|N(O_{\ell})| \leq \frac{\score(O)}{\lambda_{\min}}$. Hence,  $|N(O \cup O_\ell)| \leq \frac{2\score(O)}{\lambda_{\min}}$
\end{proof}

 Now consider the graph $G$ induced on $C_r\backslash O_{\ell}$ and $N(O \cup O_\ell )$. The number of edges incident on $C_r\backslash O_{\ell}$ is at least $(|C_r|-k)\alpha_{i}$ because for every $p \in C_r\backslash O_{\ell}$, we have $|N(O \cup O_\ell )\cap N(p)|\geq \alpha_{i}$ due to \Cref{claim:common-neighborhood}. The number of edges incident on $N(O \cup O_\ell )$ is at most $2d \cdot \score(O)/\lambda_{\min}+|C_r|^{d}(d-1)$, \Cref{claim:neighborhood-size}. This is because vertices with degrees at most $d$ can contribute $2d \cdot \score(O)/\lambda_{\min}$. Since the input graph is \kdd-free, the remaining vertices can contribute at most $\binom{|\C_r|}{d}(d-1)$ to the total number of incident edges. Using this inequality we prove the following claim.

\begin{claim}\label{claim:bounding-alpha} 
$\score(O) - \score(O_\ell) =\sum_{i=1}^{\ell} \alpha_i < \epsilon \cdot \score(O)$
\end{claim}

\begin{proof}
Consider the graph $G$ induced on $C_r\backslash O_{\ell}$ and $N(O \cup O_\ell )$. The number of edges incident on $C_r\backslash O_{\ell}$ is at least $(|C_r|-k)\alpha_{i}$ because for every $p \in C_r\backslash O_{\ell}$, we have $|N(O \cup O_\ell )\cap N(p)|\geq \alpha_{i}$ due to \Cref{claim:common-neighborhood}. The number of edges incident on $N(O \cup O_\ell )$ is at most $2d \cdot \score(O)/\lambda_{\min}+|C_r|^{d}(d-1)$, \Cref{claim:neighborhood-size}. This is because vertices with degrees at most $d$ can contribute $2d \cdot \score(O)/\lambda_{\min}$. Since the input graph is \kdd-free, the remaining vertices can contribute at most $\binom{|\C_r|}{d}(d-1)$ to the total number of incident edges.
It follows that
\[(|C_r|-k)\alpha_i\leq 2d \cdot \score(O)/\lambda_{\min}+|C_r|^{d}(d-1)\]
Taking the summation for all $i \in [\ell]$, we get 

\begin{align*}
        (|C_r|-k)\sum_{i \in [\ell]}\alpha_i &\leq 2d \ell \cdot\score(O)/\lambda_{\min}+ |C_r|^{d}\ell(d-1)\\
        \sum_{i \in [\ell]}\alpha_i &\leq \frac{2d\ell \cdot \score(O)}{(|C_r|-k)\lambda_{\min}} +\frac{\ell|C_r|^{d}(d-1)}{(|C_r|-k)}\\
        &\leq \frac{2dk\cdot \score(O)}{(|C_r|-k)\lambda_{min}} +\frac{k|C_r|^{d}(d-1)}{(|C_r|-k)}
\end{align*}

Now, since $|C_r|= r=\frac{4dk}{\epsilon\lambda_{min}}+k$, we have $\frac{2dk}{(|C_r|-k)\lambda_{min}}\leq \frac{\epsilon}{2}$. \hide{Note that $|C_r|<r$ implies $\C<r$ and hence we can return the solution by doing brute-force search respecting the running time of our algorithm.} Since, by the definition of Case 2, $\score(O) > \frac{2kr^d(d-1)}{(r-k)\epsilon}$, we have $ \frac{k|C_r|^{d}(d-1)}{(|C_r|-k)} < \frac{\epsilon\cdot\score(O)}{2}$. Consequently, we have
\[\sum_{i \in [\ell]}\alpha_i < \frac{\epsilon \cdot\score(O)}{2}+\frac{\epsilon \cdot\score(O)}{2}=\epsilon\cdot\score(O)\]
Thus $\score(O_\ell) > (1-\epsilon)\score(O)$.
\end{proof}

Thus, we have shown that there exists $O_\ell \subseteq C_r$ that is a $(1-\epsilon)$-approximate solution which proves the correctness of our algorithm. Moreover, we get the following result. 
\begin{claim}\label{lemma:highestcontsoln}
    If $t > \frac{2kr^d(d-1)}{(r-k)\epsilon}$, where $r=\frac{4dk}{\epsilon\lambda_{min}}+k$, then the set of $\ceil{r}$ vertices in $\C$ with highest $\score(\cdot)$ contains a solution with $\score(\cdot)$ at least $(1-\epsilon)t$.
\end{claim}
The running time in this case is atmost $\binom{\ceil{r}}{k}n^{\Co{O}(1)}=(\nicefrac{dk}{\epsilon})^{\mathcal{O}(k)}n^{\Co{O}(1)}$.
Thus,\hide{in both the cases} we have the desired \fptas.
\end{proof}

%% file: kernel.tex
\section{Lossy Kernel for \probonegen}

In this section, we give a kernel for our problem when the \igraph is \kddfree. Towards this, we first define the {\it optimization version} of \probonegen below.

\defparprob{\probonemax}{A bipartite graph $G = (\C, \V)$, an integer $k$, and a set $\Lambda$ of non-increasing vectors $\lambda^v=(\lambda_1^v, \lambda_2^v, \ldots, \lambda_{|A_v|}^v)$ for every $v \in \V$.}{$k$}{Find $S \subseteq \C$ such that $|S| \leq k$ and $\score _G^\Lambda(S)=\sum_{v \in \V}f_{G,v}(S)$ is maximized, where $f_{G,v}(S) = \sum_{j=1}^{|N_G(v)\cap S|}\lambda_j^v$.}
Here $\C$ represents the set of candidates and $\V$ represents the set of voters. For $v \in \V$, the set $N(v)$ represents the \approval $A_v$ of voter $v$. Let $\hat{S}\subseteq \C$ and $|\hat{S}| \leq k$ such that $\score_G^\Lambda(\hat{S})$ is maximum. Then, $\opt_{G,k,\Lambda}=\score_G^\Lambda(\hat{S})$. We also assume without loss of generality that all the voters have a nonzero \owa vector. Otherwise, if the vector corresponding to a voter is $(0,0,\ldots 0)$, we can safely delete the vertex corresponding to that voter. By ${\rm Iden}$, we denote an algorithm that outputs the input itself. We next describe some terminology.

\begin{definition}{\rm \cite{DBLP:conf/stoc/LokshtanovPRS17}}{\rm [$\alpha$-$\appa$]}\label{def:alphappa}
    Let $\alpha, \beta \in (0,1)$. An $\alpha$-{\it approximate polynomial-time pre-processing algorithm} ($\alpha$-\appa) for a parameterized optimization problem $\Pi$ is a pair of polynomial-time algorithms $\cA$ and $\cB$ called the reduction algorithm and solution lifting algorithm respectively such that the following holds:
    \begin{enumerate}
        \item given any instance $(I,k)$ of $\Pi$, $\cA$ outputs an instance $(I',k')$ of $\Pi$, and 
        \item  given any $\beta$-approximate solution of $(I',k')$, $\cB$ outputs an $\alpha\beta$-approximate solution of $(I,k)$.
    \end{enumerate}
\end{definition}
\begin{definition}{\rm \cite{DBLP:conf/stoc/LokshtanovPRS17}}{\rm [$\alpha$-approximate kernel]}
Let $\alpha \in (0,1)$. An $\alpha$-approximate kernel is an $\alpha$-\appa such that the output size $|I'| + k'$ is bounded by some computable function of $k$.
\end{definition}

\begin{definition}{\rm \cite{DBLP:journals/corr/abs-2403-06335}}{\rm [$(\alpha,\gamma)$-\appa]}\label{def:alphagammaappa}
    Let $\alpha, \gamma \in (0,1)$. An $(\alpha,\gamma)$-{\it approximate polynomial-time preprocessing algorithm}($(\alpha,\gamma)$-\appa) for a parameterized optimization problem $\Pi$ is a pair of polynomial-time algorithms $\cA$ and $\cB$ called the reduction algorithm and solution lifting algorithm respectively such that the following holds:
    \begin{enumerate}
        \item given any instance $(I, k)$ of $\Pi$, $\cA$ outputs an instance $(I', k')$ of $\Pi$, and 
        \item  given any $\beta$-approximate solution of $(I', k')$, $\cB$ outputs an $(\alpha\beta - \gamma)$-approximate solution of $(I,k)$.
    \end{enumerate}
\end{definition}

\begin{proposition}{\rm \cite{DBLP:journals/corr/abs-2403-06335}}\label{lemma:modifiedappa}
    For any $\epsilon_1, \epsilon_2, c \in (0,1)$, suppose that a maximization problem admits a polynomial-time $c$-approximation algorithm and a $(1 - \epsilon_1, \epsilon_2)$-\appa. Then, it admits a $(1 - \epsilon_1 - \epsilon_2/c)$-\appa with the same reduction algorithm.
\end{proposition}

Due to \cite{DBLP:journals/mp/NemhauserWF78}, we know that the greedy algorithm for maximizing submodular functions is a $(1-\frac{1}{e})$-approximate algorithm. The satisfaction function $f_{G,v}(\cdot)$ of each voter, $v$, is a non-decreasing submodular function \cite{DBLP:conf/aaai/SkowronF15}. Since the sum of submodular functions is also submodular, we have the following. 

\begin{lemma}\label{theorem:polyapprox}
 There is a polynomial-time $(1-\frac{1}{e})$-approximation algorithm for \probonemax.
\end{lemma}

We will split the kernel construction into two parts: first we will describe the analysis that allows us to reduce the number of candidates followed by the analysis that allows us to reduce the number of voters.

\input{reducing-candidates.tex}

\input{reducing-voters.tex}

\paragraph{Towards the kernel.} On input $(G, k, \Lambda)$, the reduction algorithm works as follows. 
\begin{enumerate}
    \item Apply $(1-\frac{\epsilon}{2})$-\appa reduction from~\Cref{lemma:reducedcandidates} to reduce the number of candidates.
        \item Apply $(1-\frac{\epsilon}{2})$-\appa reduction from~\Cref{lemma:votersreduced} to reduce the number of voters.
\end{enumerate}

The two steps ensure we get $(1-\frac{\epsilon}{2})^2$-\appa. Since $(1-\frac{\epsilon}{2})^2 \geq (1-\epsilon)$, we have $(1-\epsilon)$-\appa. The first step reduces the number of candidates to $(\frac{dk}{\epsilon})^{\mathcal{O}(d^2)}$. In the second step, the number of voters reduces to $(\frac{dk}{\epsilon})^{\mathcal{O}(d^3)}$. Consequently, we obtain the following result.

\begin{theorem}\label{thm:lossy-kernel}
    For any $d \in \mathbb{N}$ and $\epsilon \in (0,1)$, there is a parameter preserving $(1-\epsilon)$-approximate kernel for \probonemax when the \igraph is \kddfree with $(\frac{dk}{\epsilon})^{\mathcal{O}(d^2)}$ candidates and  $(\frac{dk}{\epsilon})^{\mathcal{O}(d^3)}$ voters.
\end{theorem}

%% file: reducing-candidates.tex
\paragraph{Reducing the number of candidates}

\begin{lemma}\label{lemma:reducingcandidate}
Suppose that $\cA$ is a parameter-preserving reduction algorithm for \probonemax that on input $\Co{I}= (G,k, \Lambda)$ just deletes a subset of candidates resulting in the instance $\Co{I'}=(G', k, \Lambda)$. If $\opt_{G', k, \Lambda} \geq (1 -\delta)\cdot\opt_{G, k, \Lambda}$ for some $\delta \in (0, 1)$, then $(\cA, \Iden)$ is a $(1-\delta)$-\appa.
\end{lemma}

\begin{proof}
    Consider any $\beta$-approximate solution $Y$ to $\Co{I'}$. Since $G'$ results from deleting vertices from $\C$ (candidates), we have $\score_G(Y)\geq \score_{G'}(Y)\geq \beta \opt_{G', k, \Lambda} \geq \beta(1-\delta)\cdot\opt_{G, k, \Lambda}$. Thus, Definition~\ref{def:alphappa} implies that $(\cA, \Iden)$ is a $(1-\delta)$-\appa.
\end{proof}

\begin{lemma}\label{lemma:reducedcandidates} For any $\epsilon \in (0,1)$, there is a parameter-preserving $(1-\epsilon)$-\appa for \probonemax when the \igraph is \kddfree, such that the output has $(\frac{dk}{\epsilon})^{\mathcal{O}(d^2)}$ candidates.
\end{lemma}

\begin{proof}
We apply the polynomial time $(1-\frac{1}{e})$-approximation algorithm from~\Cref{theorem:polyapprox}. Let $\optapx$ denote the \score of the returned solution and $r=\frac{4dk}{\epsilon\lambda_{min}}+k$. We have the following two cases:

\begin{description}[wide=0pt]
\item[Case 1:] $\optapx > \frac{2kr^d(d-1)}{(r-k)\epsilon}$. In this case, we delete all but the $r$ highest degree vertices in $\C$. We will show that $\opt_{G', k, \Lambda}\geq (1-\epsilon)\opt_{G, k, \Lambda}$, then, by~\Cref{lemma:reducingcandidate} it follows that we have a $(1-\epsilon)$-\appa.
We have
$\opt_{G, k, \Lambda}\geq \optapx > \frac{2kr^d(d-1)}{(r-k)\epsilon}$. Next, using~\Cref{lemma:highestcontsoln}, we have $\opt_{G',k, \Lambda} > (1-\epsilon)\opt_{G,k,\Lambda}$ and we are done.

\item[Case 2:] $\optapx \leq \frac{2kr^d(d-1)}{(r-k)\epsilon}$. In this case, we have $\opt_{G,k,\Lambda}\leq \frac{e}{e-1}\optapx \leq \frac{e}{e-1} \frac{2kr^d(d-1)}{(r-k)\epsilon}$. Let $\psi=\ceil{\frac{e}{e-1} \frac{2kr^d(d-1)}{(r-k)\epsilon}}$. We have for each $v \in \C$, $deg(v)< \nicefrac{\psi}{\lambda_{\min}}$; otherwise, $v$ itself is a solution. We apply \Cref{redrule:sunflower} exhaustively with $Wk+1 = (\nicefrac{\psi}{\lambda_{\min}})k+1$. Let the final graph be $G'=(\C',\V')$ where $\V'$ is obtained by deleting all the isolated vertices. By~\Cref{lemma:kabfreesunflower}, we know that there are at most $d((\nicefrac{\psi}{\lambda_{\min}})^2k)^d=(\frac{dk}{\epsilon})^{\mathcal{O}(d^2)}$ vertices in $\C'$. Now, we need to show that this is a $(1-\epsilon)$-\appa. By ~\Cref{lemma:optsame}, we know that the application of~\Cref{redrule:sunflower} does not change the optimum $\score(\cdot)$ value. Thus, we have $\opt_{G',k, \Lambda}\geq \opt_{G,k, \Lambda}$, which together with~\Cref{lemma:reducingcandidate}, implies that the reduction algorithm in \textbf{Case 2} is also a $(1-\epsilon)$-\appa.
\end{description} 
This completes the proof of the lemma.
\end{proof}

%% file: reducing-voters.tex
\paragraph{Reducing the number of voters}
\begin{lemma}\label{lemma:reducingvoters}
    Suppose that $\Co{A}$ is a parameter-preserving reduction algorithm for \probonemax that also preserves the set of candidates and \owa vectors, i.e, on input $\Co{I}=(G = (\C, \V), k, \Lambda)$, it produces $\Co{I'}=(G' = (\C, \V'), k, \Lambda)$. If there exists $\delta, h \geq 0$ and $s > 0$ (where $h$ and $s$ can depend on \Co{I}) s.t. the following holds for any $k$-sized subset $X\sse \C$: 
 \begin{equation}\label{eqn:lemreducevoter}
         |\score_G(X)-s\cdot \score_{G'}(X)-h|\leq \delta \cdot \opt_{\Co{I}}, 
    \end{equation}
then, for every $\delta_1 \in (0,1)$, $ (\cA, \Iden)$ is a $(1-\delta_1, 2\delta)$-\appa. 
\end{lemma}

\begin{proof} 
Let $Y^*$ denote an optimum solution for \Co{I}. Let $Y$ denote a $\beta$-approximate solution for \Co{I'}. We note that $\opt_{G,k,\Lambda}= \score_{G}(Y^*)$. We observe that \Cref{eqn:lemreducevoter} implies $\score_G(Y) \geq s\cdot\score_{G'}(Y)+h -\delta\cdot \score_{G}(Y^*)$. We argue as follows. 
\begin{align*}
            \score_G(Y) &\geq s \cdot \score_{G'}(Y) +h -\delta\cdot \score_{G}(Y^*)\\ 
            &\geq s\beta \cdot \opt_{\Co{I'}} +h -\delta \cdot \score_{G}(Y^*)\\
            &[\text{since $Y$ is a $\beta$-approximate solution for \Co{I'}.}]\\
            &\geq s\beta \cdot \score_{G'}(Y^*) +h -\delta \cdot\score_{G}(Y^*)\\
                & \geq s\beta\cdot \big(\frac{1}{s}(\score_G(Y^*)-h-\delta\cdot\opt_{\Co{I}})\big)+h -\delta \cdot \score_{G}(Y^*)\\
            &\geq \beta(\score_G(Y^*)-h-\delta\cdot\opt_{\Co{I}})+h-\delta \cdot \opt_{\Co{I}}\\
            &\geq \beta \opt_{\Co{I}}-\beta h-\beta\delta\opt_{\Co{I}}+h-\delta \cdot \opt_{\Co{I}}\\
            &\geq (\beta -2\delta)\cdot \opt_{G,k,\Lambda}
    \end{align*}
Now $(\beta -2\delta)\cdot \opt_{G,k,\Lambda}\geq ((1-\delta_1)\beta-2\delta)\cdot \opt_{G,k,\Lambda} $ for any $\delta_1 \in (0,1)$. Thus, due to Definition~\ref{def:alphagammaappa} we have that for any $\delta_1 \in (0,1)$, $(\cA, \Iden)$ is a $(1-\delta_1, 2\delta)$-\appa.

\end{proof}

\begin{lemma}\label{lemma:votersreduced}
     For any $\epsilon \in (0,1)$, there is a parameter-preserving $(1-\epsilon)$-\appa for \probonemax when the \igraph is \kddfree, such that the output graph has the same set of candidates and $\mathcal{O}(k\cdot d \cdot n^{d+1}/\epsilon)$ voters.
\end{lemma}
\begin{proof}
We want to prove the result for any $\epsilon \in (0,1)$. To proceed, without loss of generality, we fix an arbitrary $\epsilon \in (0,1)$. Let $0<\Tilde{\epsilon}<\epsilon$ and $\epsilon^*=(1-\frac{1}{e})\Tilde{\epsilon}$. 

Suppose that we have a $(1-\delta_2, \epsilon^*)$-\appa for any $\delta_2\in (0,1)$. Then due to~\Cref{lemma:modifiedappa} and~\Cref{theorem:polyapprox} we have $(1-\delta_2-\frac{\epsilon^*}{1-\frac{1}{e}})$-\appa which is a $(1-\delta_2-\Tilde{\epsilon})$-\appa. Since $\delta_2$ can take any value in $(0,1)$ we set $\delta_2= \epsilon -\Tilde{\epsilon}$ to get $(1-\epsilon)$-\appa

Thus, to prove this lemma, it is sufficient to show that for any $\delta_2\in (0,1)$, we have a $(1-\delta_2,\epsilon^*)$-\appa. On input $\Co{I}=(G,k, \Lambda)$, the reduction algorithm works as follows.
\begin{enumerate}
    \item  Use~\Cref{theorem:polyapprox} to compute $\optapx$ such that $\opt_{\Co{I}}\geq \optapx\geq(1-\frac{1}{e})\opt_{\Co{I}}$. Let $s=\frac{\epsilon^*\optapx}{k\cdot10dn^d}$.
    \item Let $\V_{set}$ denote the subset of vertices in $\V$ with distinct neighborhoods, i.e., the set of voters with distinct \approval. 
    \item We start with $\V'$ being an empty multiset. For each $v \in \V_{set}$, let $m_v$ denote the number of occurrences of $v$ in $\V$. We add $\floor{\nicefrac{m_v}{s}}$ copies of $v$ to $\V'$. We define graph $G'=(\C,\V')$.
    \item We output $\Co{I'}=(G',k, \Lambda)$. 
\end{enumerate}
Since the degree of every vertex in $\C$ is at most $\nicefrac{\opt_{\Co{I}}}{\lambda_{\min}}$, we have $|\V|\leq n\opt_{\Co{I}}/\lambda_{\min}$. Thus, by definition of $\V'$, 

\[|\V'|\leq \frac{|\V|}{s}\leq \hide{\frac{ n\opt_{\Co{I}}}{\lambda_{\min}s}=} \left(\frac{10\cdot\opt_{\Co{I}}}{\optapx\lambda_{\min}}\right)\frac{kdn^{d+1}}{\epsilon^*}=\mathcal{O}\left(\frac{kdn^{d+1}}{\epsilon}\right)\]

We claim that for every $k$-sized subset $Y\subseteq \C$, we have $|\score_G(Y)-s\cdot \score_{G'}(Y)|\leq \frac{\epsilon^*}{2}\opt_{\Co{I}}$. This with ~\Cref{lemma:reducingvoters} yields that $(\cA,\Iden)$ is a $(1-\delta_2,\epsilon^*)$-\appa for any $\delta_2\in(0,1)$ as desired.

To see that the claim holds, we observe that 
\begin{align*}
    &|\score_G(Y) - s \cdot \score_{G'}(Y)| \leq \sum_{v \in \V_{set}}k\left|m_v -s \cdot \Ceil{\frac{m_v}{s}}\right|\lambda_1^v\\
    &\leq ks\cdot |\V_{set}|, \text{ since $\lambda_1^v\leq1$, for each $v\in \V$.}
\end{align*}
Since $G$ is \kddfree, therefore for every $d$-sized subset in $\C$ there can be at most $d$ common neighbors in $\V$. Thus the number of vertices in $\V$ with degree at least $d$ is at most $dn^d$. The number of vertices with unique neighborhood and with degree at most $d$ is $n^d$. Thus, we have
\begin{align*}
    &|\score_G(Y)-\!s\cdot \score_{G'}(Y)|\leq ks \cdot |\V_{set}|\leq ks(dn^d+n^d)\\
    &= k\left(\frac{\epsilon^*\optapx}{10kdn^d}\right)(d+1)n^d\leq \!\frac{\epsilon^*(d+1)\opt_{\Co{I}}}{10d}\leq\! \frac{\epsilon^*}{2}\opt_{\Co{I}} 
\end{align*}
\end{proof}

%% file: additiveapprox.tex
\section{Additive Parameterized Approximation}
In this section, we design a one-additive parameterized approximation algorithm. 
In particular, we achieve the following: given a \yes-instance  ${\cal I}=(G,k,t,\Lambda)$ of \probone, where $G$ is a $K_{d,d}$-free graph, we output a committee of size $k+1$ whose score is at least $t$, in time \fpt in $k+d+\epsilon$. Note that, we may return a $(k+1)$-sized committee even for a \no-instance. But, if the algorithm returns \no, then ${\cal I}$ is a \no-instance of \probone. 

 We first give an intuitive description of the algorithm.  
 If the candidate set $\C$ is bounded by $g(k,d)$, then we can try all possible subsets to obtain a solution to ${\cal I}$. 
 If $t$ is bounded by $f(k,d)$, then observe that the degree of every vertex in $\C$ is bounded by $\nicefrac{f(k,d)}{\lambda_{\min}}$, where $\lambda_{\min} \leq 1$ is a constant; otherwise, a vertex of the highest degree is a solution to ${\cal I}$. So, we apply ~\Cref{redrule:sunflower} with appropriately chosen $W$ and bound the size of $\C$ by another function of $k+d$, and now again we can try all possible subsets of $\C$. When none of the above cases hold, we either correctly return \no or for a \yes-instance, we  find a committee of size $k$, say $S'$, using \apxvot (\Cref{alg:combgenapxthn}) for $\epsilon = \nicefrac{\lambda_{\min}}{4k}$ (the choice of $\epsilon$ will be clear later). Recall that \apxvot returns a $(1-\epsilon)$-approximate solution, thus, $\score_G^\Lambda(S')\geq (1-\epsilon)t$. We construct a large enough set of candidates (bounded by a function of $(k,d)$) of high score, say $H$, and argue that, given a \yes-instance, either every solution to ${\cal I}$ contains a vertex from $H$, or there is a vertex $x$ in $H$ such that $\score_G^\Lambda(S'\cup \{x\})\geq t$. 

\Cref{alg:genapxad} describes the procedure formally. For $x\in \C$, let us recall the definition of $\Lambda_x$. 
Firstly, for any \owa vector $\lambda^v =\{ \lambda^v_1,\lambda^v_2,\ldots, \lambda^v_{|A_v|}\}$, where $v\in \V$, let  $\lambda^v_{-1}$  denote the \owa vector starting from the second entry, i.e., $\{ \lambda^v_2,\lambda^v_3,\ldots, \lambda^v_{|A_v|}\}$. Then $\Lambda_x$ is the set $\cup_{v\in N(x)}\{ \lambda^v_{-1}\}\cup_{v \in B\backslash N(x)}\{ \lambda^v$\}, i.e., we delete the first entry of the \owa vectors of neighbors of $x$, and rest remains the same. Also, let $\V_0$ and  $\V_\emptyset$ denote the voters with all-zero vectors and empty vectors, respectively. We assume that our input instance does not contain any such voters.

We prove the correctness of \Cref{alg:genapxad} in the following lemma.
\begin{lemma} \label{lem:additive-apx-correctness}
  Given a \yes-instance $(G,\Lambda,k,t)$ of \probone,    \Cref{alg:genapxad} returns a \solution of size at most $k+1$ whose score is at least $t$. 
\end{lemma}
\begin{proof}
     Let ${\cal I}=(G,\Lambda,k,t)$ be a \yes-instance of \probone. We prove the correctness by induction on $k$.

   \noindent \emph{Base Case}: $k=0$. Since ${\cal I}$ is a \yes-instance, $t\leq 0$. Thus, empty set is a solution to ${\cal I}$ as returned by the algorithm. \\
\noindent \emph{Induction Step}: Suppose that the claim is true for all $i\leq k$.  Next, we argue for $k=i+1$.  If Step~\ref{line no:addappcase2} or \ref{line no:addappcase1} is executed, then since we try all possible subsets, for a \yes-instance, we return a set of size $k$. If the condition in Step~\ref{line no:FPT_-APX} is executed and we return a set in Step~\ref{line no:addappcase3}, then we return a set of size at most $(k+1)$ whose score is at least $t$. Suppose that we execute Step~\ref{line no:addappcase4}, then we first claim that for a  \yes-instance, ${\cal I}$, one of the instances in Step~\ref{line no:addappcase5} is a \yes-instance. Let $S$ be a solution to ${\cal I}$. We first note that since ${\cal I}$ is a \yes-instance, the condition in Step~\ref{line no:FPT_-APX} is true due to \Cref{thm:fpt-apx}. But, since we are executing Step~\ref{line no:addappcase4}, we did not find a desired set in Step~\ref{line no:addappcase3}. Thus, due to \Cref{lemma:oneadditive}, $S\cap H\neq \emptyset$. Suppose $y\in S \cap H$. Let $S'=S\setminus y$. Let $G_y=(G-y) - (\V_0 \cup \V_\emptyset)$. Then, $S'$ is a solution to ${\cal I}_y=(G_y,k-1,t-\score_{G}^\Lambda(y),\Lambda_y)$. Hence, due to induction hypothesis, there exists a $k$-sized subset of candidates $S_y$ such that  $\score_{G_y}^{\Lambda_y}(S_y)\geq t-\score_{G}^\Lambda(y)$.  Next, we argue that $\score_G^\Lambda(S_y \cup y)\geq t$.  

\begin{align*}
\score_G^\Lambda(S_y\cup y)=&\sum_{v\in \V}\sum_{i=1}^{|(S_y\cup y)\cap N(v)|}\lambda_i^v \\
=& \sum_{v \in N(y)}\lambda_1^v+\sum_{v \in N(y)}\sum_{i=2}^{|(S_y\cup y)\cap N(v)|}\lambda_i^v\\
&+\sum_{v\in \V \backslash N(y)}\sum_{i=1}^{|(S_y\cup\{y\})\cap N(v)|}\lambda_i^v
\end{align*}
Since
\begin{align*}
    \score_{G_y}^{\Lambda_y}(S_y) & =\sum_{v \in N(y)}\sum_{i=2}^{|(S_y\cup y)\cap N(v)|}\lambda_i^v \\
    & +\sum_{v\in \V \backslash N(y)}\sum_{i=1}^{|(S_y\cup\{y\})\cap N(v)|}\lambda_i^v
\end{align*}
In the above inequality, we considered $\lambda$ vectors in $\Lambda$.
Thus, we have that
\begin{align*}
\score_G^\Lambda(S_y\cup y)= & \score_{G}^\Lambda(y)+\score_{G_y}^{\Lambda_y}(S_y) \geq t.
\end{align*}
This completes the proof.
\end{proof}
\begin{algorithm}[tbh]
\caption{\addapxth : An \fpt algorithm for \oneadapp of \probonemax}\label{alg:genapxad}
\textbf{Input:} A bipartite graph $G=(\C,\V,E)$, a set $\Lambda =\{\lambda^v:v\in \V\}$, and non-negative integers $k$ and $t$.\\
\textbf{Output:} Either a set $S\subseteq \C$ s.t. $|S|\leq k+1 $ and $\score_G^\Lambda(S)\geq t$, or ``\no".
\begin{algorithmic}[1]
\If {$k=0$, $t\leq 0$ }
\Return an empty set. 
\EndIf
\If {$k=0$, $t> 0$ }
\Return \no
\EndIf

\If {$|\C|\leq k(d-1)(4k^2)^{d-1}+1$} \label{line no:addappcase2}
\If {there exists a $k$-sized set $S\subseteq \C$ s.t. $\score_G^\Lambda(S)\geq t$} 
\Return $S$ 

\Else{ \Return \no} 
\EndIf
\EndIf
\If  {$t \leq  8k^4d\lambda_{min}$} \label{line no:addappcase1}
\State apply Reduction Rule~\ref{redrule:sunflower} exhaustively with $W=\frac{t}{\lambda_{\min}}$, 

\If {there exists a set $S\subseteq \C$ s.t. $\score_G^\Lambda(S)\geq t$}  
\Return $S$

\Else{ \Return \no}

\EndIf
\EndIf

\If {\apxvot$(G,k,t,\epsilon=\frac{\lambda_{\min}}{4k},\Lambda)$ returns a set $S'$}
\label{line no:FPT_-APX}
\hide{\State Let $S'= \apxvot(G,k,t,\epsilon=\frac{\lambda_{\min}}{4k},\Lambda)$}
\State Let $H \subseteq \C$ be a set of $k(d-1)(4k^2\lambda_{\min})^{d-1}+1$ candidates of highest  score.
\If {there exists $x \in H$ such that $\score_G^\Lambda(\{x\}\cup S')\geq t$} \label{line no:addappcase3}
\Return $S' \cup \{x\}$
\EndIf
\Else \For{$y \in H$}\label{line no:addappcase4}
\State let $G_y=(G \sm \{y\}) \sm (\V_{0} \cup \V_{\emptyset})$

\hide{\COMMENT{{\color{blue} after deleting $y$, \owa vectors of some of the voters might become empty or zero}.}}

\If{\addapxth$(G_y,k-1,t-\score_{G}^\Lambda(y),\Lambda_y)$ returns a set $S$}\label{line no:addappcase5} 
\Return $S \cup \{y\}$
\EndIf
\EndFor
\EndIf
\end{algorithmic}
\end{algorithm}

For proving \Cref{lem:additive-apx-correctness}, we establish a crucial result (in \Cref{lemma:oneadditive}) that forms the core of our algorithm. 
Towards this, we first
define a notion of \rm{High Degree Set} as follows.

\hide{
\il{Why restating the problem }
We restate the problem again for easier readability. 
\defparprob{\probonegen}{A bipartite graph $G=(A,B)$, integer $k$, $t$ and a set, $\Lambda$, of  non-increasing vectors $\lambda^v=(\lambda_1^v,\lambda_2^v,\cdots \lambda_{deg(v)}^v)$ for all $v\in B$}{$k$}{Does there exist $S\subseteq A$ such that $|S|\leq k $ and $\score (S)=\sum_{v \in B}u_{G,v}(S)\geq t$ where $u_{G,v}(S)=\sum_{j=1}^{|N_G(v)\cap S|}\lambda_j^v$ ?}
}

\begin{definition}{\rm \cite{DBLP:conf/soda/0001KPSS0U23}}{\rm [$\beta$-High Degree Set]} \label{def:def3} 
Given a bipartite graph $G=(A,B,E)$, a set $X\subseteq B$, and a positive integer $\beta > 1$, the {\em $\beta$-High Degree Set}, is defined as: 

\hide{denoted by ${\sf HD}_{\beta}^{G}(X) \subseteq A$, is a set of vertices s.t. for each vertex $v \in {\sf HD}_{\beta}^{G}(X) $ satisfies $|N(v)\cap X| \geq \frac{|X|}{\beta}$ and $|N(v)| \geq d$, i.e.,}
\begin{center}
$ {\sf HD}_{\beta}^{G}(X) = \{v: v\in A, |N(v)|\geq d, |N(v)\cap X| \geq \frac{|X|}{\beta}\}$
\end{center}
\end{definition}

Interestingly, the size of $\beta$-High Degree Set is bounded for $K_{d,d}$-free bipartite graphs as shown by the following.

\begin{proposition}\label{lemma:hibound}{\rm \cite{DBLP:conf/soda/0001KPSS0U23}}
For all $d$ and for all $\beta > 1$ with $\frac{|X|}{2\beta} > d$, where $X\subseteq B$, if $G = (A, B, E)$ is \kddfree\hide{bipartite graph}, then 
$|{\sf HD}_{\beta}^{G}{(X)}| \leq f(\beta, d) = (d-1)(2\beta)^{d-1}$.
\end{proposition}

Next, we move towards proving the core result for our algorithm.

\begin{lemma}\label{lemma:oneadditive}
Let $(G,\Lambda,k,t)$ be an instance of \probone, where $G$ is a \kddfree graph with $V(G) = \C \uplus \V$. Let $\ell \leq k$ and $t'\leq t$ be two positive integers. Let $S'\subseteq \C$ be an $\ell$-sized set such that $\score_G^\Lambda(S')\geq t'(1-\frac{\lambda_{\min}}{4\ell})$, where $t'\geq 8\ell^4d\lambda_{\min}$,  and $H$ be a set of $\ell (d-1)(4\ell^2\lambda_{\min})^{d-1}+1$ highest score candidates in $\C$. For any $S\subseteq \C$ of size $\ell$ with $\score_G^\Lambda(S)\geq t'$, either $S\cap H\neq\emptyset$ or there exists a candidate $ x \in H$ such that $\score_G^\Lambda(\{x\}\cup S') \geq t'$. If there exists a vertex $v$ of degree at least $\frac{t'}{\lambda_{\min}}$, then $S'=\{v\}$. 
\end{lemma}
\begin{proof}
Suppose that there exists a candidate $x\in H$ whose score is at most $\frac{t'}{\ell+1}$. Then, every candidate in $\C \setminus H$ has score at most $\frac{t'}{\ell+1}$. In this case, we claim that $S\cap H\neq \emptyset$. Suppose not, then $S\subseteq \C \setminus H$, and hence $\score_G^\Lambda(S) \leq \frac{\ell t'}{\ell+1} < t'$, a contradiction. Thus, in this case $S\cap H \neq \emptyset$. Next, we consider that the score of every candidate in $H$ is more than $\frac{t'}{\ell+1}$. In this case, we will show that there exists a candidate $x \in H$ such that $\score_G^\Lambda(\{x\}\cup S') \geq t'$. 
Towards the contradiction, suppose that for every $y \in H$, $\score_G(\{y\}\cup S') < t'$. Clearly, due to the lemma statement, the degree of every vertex is at most $\nicefrac{t'}{\lambda_{\min}}$. Since $\score_G^\Lambda(S')\geq t'(1-\frac{\lambda_{\min}}{4\ell})$, for all $y\in H$, $|N(y)\backslash N(S')|<\frac{t'\lambda_{\min}}{4\ell\lambda_{\min}}=\frac{t'}{4\ell}$. We next argue that every $y\in H$ has a large neighborhood in $S'$.  
\begin{align*}
|N(y)\cap N(S')| &= |N(y)|-|N(y)\backslash N(S')|
\end{align*}
Recall that $\score_G^\Lambda(y)> \frac{t'}{\ell+1}$, and $\lambda^v$ is a non-increasing vector with $\lambda^v_1 \leq 1$.  Thus, $|N(y)|> \frac{t'}{\ell+1}$. Hence, 
\begin{align*}
 |N(y)\cap N(S')| &\geq  
 \frac{t'}{\ell+1}-\frac{t'}{4\ell}\\& =t'\left(\frac{1}
{\ell+1}-\frac{1}{4\ell}\right)\\& \geq \frac{t'}{2\ell}.
\end{align*}
Thus, using the pigeonhole principle, for every $y\in H$, there exists a candidate $u \in S'$ such that 
\begin{align}
    |N(y)\cap N(u)|\geq \frac{|N(y)\cap N(S')|}{|S'|}\geq \frac{t'}{2\ell^2}
\label{eq:add-apx-large-intersection}\end{align}

Again, using the pigeonhole principle, we know that there exists a candidate $u \in S'$ such that there are at least $\frac{|H|}{\ell}$ many candidates $y\in H$, with $|N(y)\cap N(u)|\geq \frac{t'}{2\ell^2}$. We denote all these vertices by $H_u$. Consequently, we have $|H| \leq \ell|H_u|$.

Next, we construct a bipartite graph $G_u=G[H_u\cup N(u)]$.  

Consider $\beta=2\ell^2\lambda_{\min}$. Since $H_u \subseteq H$, for every vertex $y\in H_u$, $\score_G^\Lambda(y)\geq \frac{t'}{\ell+1}$. Since $t' \geq 4\ell^2d$, $\score_G^\Lambda(y)\geq d$. Note that $|N(y)|\geq \score_G^\Lambda(y)$ as $\lambda_{\min}\leq 1$. Thus, $|N(y)|\geq d$. Recall that the degree of every vertex is at most $\nicefrac{t'}{\lambda_{\min}}$. Thus, $|N(u)|\leq \nicefrac{t'}{\lambda_{\min}}$. Hence,
\begin{align*}
    |N(y)\cap N(u)| \geq \frac{t'}{2\ell^2} \geq \frac{\lambda_{\min}|N(u)|}{2\ell^2}=\frac{|N(u)|}{\beta} 
\end{align*}
We can also assume $\frac{|N(u)|}{2\beta}>d$, otherwise $|N(u)|\leq 2\beta d\leq 4\ell^2\lambda_{\min}d$. Then, for each vertex $y\in H_u$, $4k^2\lambda_{min}d\geq  |N(y)\cap N(u)| \geq \frac{t'}{2\ell^2}$ which implies $t'\leq 8k^4\lambda_{\min}d$ which is a contradiction to our assumption of $t'$

Thus, due to \Cref{def:def3}, $H_u \subseteq {\sf HD}_{\beta}^{G_u}{(N(u))}$. By applying~\Cref{lemma:hibound} on $N(u)$, we get $|H_u|\leq |{\sf HD}_{\beta}^{G_u}{(N(u))}| \leq (d-1)(4\ell^2\lambda_{\min})^{d-1}$. 

Recall that $|H| \leq \ell|H_u|$. Hence, we have $|H| \leq \ell (d-1)(4\ell^2\lambda_{\min})^{d-1}$. But this contradicts the definition that $|H|\geq \ell (d-1)(4\ell^2\lambda_{\min})^{d-1}+1$.

\end{proof}

\textbf{Running Time:} The running time of the algorithm is governed by the following recurrence relation
\begin{enumerate}
\item $T(k) \leq  (k (d-1)(4k^2\lambda_{min})^{d-1}+1) \cdot T(k-1)+ (kd)^{\OO(kd)}+ (k (d-1)(4k^2\lambda_{min})^{d-1}+1)^{k+1}n^{\mathcal{O}(1)}+\left(\frac{dk}{\epsilon}\right)^{\OO(d^2k)} n^{\OO(1)}$ (where $\epsilon=\frac{\lambda_{min}}{4k}$).
\item $T(0)=n^{\mathcal{O}(1)}$.
\end{enumerate}
This is because in the first three cases the algorithm takes time $(kd)^{\OO(kd)}$, $(k (d-1)(4k^2\lambda_{min})^{d-1}+1)^{k+1}n^{\mathcal{O}(1)}$, and $(dk)^{\OO(d^2k)} n^{\OO(1)}$ respectively. Solving the recurrence, we get $T(k)\leq (dk)^{\OO(d^2k)}n^{\OO(1)}$.

\begin{theorem}\label{theorem:addapprox}
    There exists an algorithm for \probonegen that runs in time $(dk)^{\OO(d^2k)}n^{\OO(1)}$, and returns a set $S \subseteq A$ of size at most $k+1$ such that $\score_G(S)\geq t$. 
\end{theorem}

%% file: parabyt.tex
\section{Parameterized by Threshold}\label{sec:parabyt}
For the sake of clarity we restate our problem in the \owa framework. Note that we provide an algorithm for the case when every voter has same \owa vector. For more details regarding equivalence with the \owa framework, refer to~\Cref{section:owaframework}.
\defparprob{\probonegen}{A bipartite graph $G=(\C,\V,E)$, a non-increasing vector $\lambda=(\lambda_1,\lambda_2,\ldots, \lambda_{k})$, and positive integers $k$ and $t$.}{$k$}{Does there exist $S\subseteq \C$ such that $|S|\leq k $ and $\score_G(S)=\sum_{v \in \V}f_{G,v}(S)\geq t$ where $f_{G,v}(S)=\sum_{j=1}^{|N_G(v)\cap S|}\lambda_j$ ? } 
For the special case of PAV we have $\lambda=\{1,\frac{1}{2},\ldots \frac{1}{k}\}$. We also know that $\lambda_1\leq 1$ (\Cref{section:owaframework})
\begin{lemma}
    \probonegen admits an \fpt algorithm parameterized by $t+k$. 
\end{lemma}

\begin{proof}
We first provide a randomized algorithm and then show that it can be easily derandomized using standard techniques.
    Consider a solution committee $O=\{o_1,o_2,\ldots o_k\}\subseteq \C$ consisting of $k$ candidates. For convenience we assume $\score_G(O)=t$. Let $t'$ be the number of voters $v$, with $f_{G,v}(O)\geq 0$. Note that $t' \leq t/\lambda_1$, hence we guess $t'$. Let $v_1,v_2,\ldots v_t'$ be those voters. In the profile graph we color the elements of $\C$ with $k$ colors and the elements of $\V$ with $t'$ colors. Let $Y[p,q]$ denote the set of all possible bipartite graphs $G=(A,B)$ where $A=\{a_1, \ldots a_p\}$ and $B=\{b_1\ldots b_q\}$. There are $2^{pq}$ such graphs since we have $2^q$ possible neighbors for each vertex in $A$. We construct $Y[k,t']$ and for every $g \in Y[k,t']$ we choose a candidate from each of the $k$ color classes with at least one neighbor in each of the color classes corresponding to its neighbor in $g$. That is for color class $i$ we choose a vertex that has neighbors in the color classes of the set $\{j \mid (a_i,b_j) \in E(g)\}$. Let $S$ be the set of chosen candidates. If $\score_G(S)\geq t$ then we return it as a solution.  A description of the algorithm is provided in Algorithm~\ref{alg:parabyt}.

\begin{algorithm}[h]
\caption{An \fpt Algorithm for \probonegen parameterized by threshold ($t$).}\label{alg:parabyt}
\textbf{Input:} A bipartite graph $G=(\C,\V,E)$, a non-increasing vector  $\lambda=(\lambda_1,\lambda_2,\ldots, \lambda_{k})$, and positive integers $k$ and $t$ \\
\textbf{Output:} A $k$-sized subset $S\subseteq \C$ such that \hide{$|S|\leq k $ and} $\score_G(S)\geq t$.

\begin{algorithmic}[1]
\State $i=0$
\For{$t' \in [t,\nicefrac{t}{\lambda_1}]$}
\State For every $c \in \C$ u.a.r assign a color from $[k]$.
\State For every $v \in \V$ u.a.r assign a color from $[t']$
\State Let $S=\emptyset$
\For{$g \in Y[k,t']$}
\For{ $i \in [k]$}
\For{$c \in \C$ such that $c$ is assigned color $i$}
\If{ $c$ has neighbors in color classes of the set  $\{j \mid (c,b_j) \in E(g)\}$ }
\State $S=S \cup \{c\}$
\State Break
\EndIf
\EndFor
\EndFor
\If{$\score_G(S)\geq t$}
\Return $S$
\EndIf
\EndFor
\EndFor
\end{algorithmic}
\end{algorithm}

We say coloring is a good coloring if $\forall i \in [k], o_i$ gets color $i$ and $\forall i \in[t'], v_i$ gets color $i$. The probability that $o_i$ gets color $i$ is $\frac{1}{k}$ and the probability that $v_i$ gets color $i$ is $\frac{1}{t'}$. Hence, the probability of good coloring is  $(\frac{1}{k})^k(\frac{1}{t'})^{t'}$. Each candidate $o_i$ is adjacent to a subset of $\{v_1,\ldots v_{t'}\}$. Let the subset be $J_i= \{v_j\mid (o_i,v_j) \in E(G)\}$ . Consider the graph induced by $O$ and the voter set $\{v_1, \ldots v_{t'}\}$. This graph, say $g$,  appears in the set $Y[k,t']$. Consider the case when we choose candidates corresponding to $g$. 
Now in case of a \yes instance with a good coloring we can choose a candidate corresponding to each $o_i $ from color class $i$ with neighbors in the color classes corresponding to colors of vertices in $J_i$. Now we show that if we are able to choose $k$ such candidates, say $S$, then $\score_G(S)\geq t$. Let $i \in [t]$ be color class of voters. Suppose $v_i$ approves $j$ candidates in $O$ then its contribution to the total score is $\sum_{a=1}^j\lambda_a $. By our algorithm, now there are $j$ candidates in $S$ which have neighbors in the color class $[i]$. The voters in color class $i$ now contribute at least $\sum_{a=1}^j\lambda_a $ since $\lambda_a$'s are non-increasing. The score contributed by each voter $v_i$ is contributed by the voters in color class $i$, hence $\score_G(S)\geq \score_G(O)\geq t$. Thus our algorithm runs in time $t'2^{kt'}n^{\OO(1)}$ and returns a solution with probability  $(\frac{1}{k})^k(\frac{1}{t'})^{t'}$. We boost the success probability to a constant factor by repeating the algorithm $k^kt'^{t'}$ times. Thus the overall running time is $k^kt'^{t'+1}2^{kt'}n^{\OO(1)}$. \\

Now we demonstrate how this algorithm implies an \fpt algorithm parameterized by $t$ under the PAV rule.

If $k \leq t$, then we directly obtain an algorithm parameterized by $t$. Now consider the case where $k \leq \delta$, where $\delta$ is the maximum degree of a vertex in $\V$. In this case, if $t \geq \sum_{i=1}^k \frac{1}{i}$, then we have $k \leq 2^t$, which again implies an \fpt algorithm with parameter $t$. Otherwise, if $t \leq \sum_{i=1}^k \frac{1}{i}$, then selecting any $k$ neighbors of the highest-degree vertex in the committee yields a valid solution, which can be found in polynomial time.

Finally, we consider the case where $k > \delta$. If $t \leq \sum_{i=1}^\delta \frac{1}{i}$, then choosing all the neighbors of the highest-degree vertex in the committee gives the desired solution in polynomial time. Otherwise, $t > \sum_{i=1}^\delta \frac{1}{i}$ implies $\delta < 2^t$.

Let $O = \{o_1, o_2, \dots, o_k\}$ be the $k$ candidates selected, and let $V = \{v_1, \dots, v_{t'}\}$, where $t' < t$, denote the voters receiving positive satisfaction (i.e., each has at least one neighbor in $O$). Thus, $O \subseteq N(V)$. Since $|N(V)| \leq |V|\delta$, we have $|O| \leq |V|\delta \leq t \cdot 2^t$, which implies $k \leq t \cdot 2^t$. Hence, in this case as well, we obtain an \fpt algorithm parameterized by $t$.

\textbf{Derandomization:}
The algorithm can be derandomized using standard techniques \cite{DBLP:books/sp/CyganFKLMPPS15}. In particular, we use an $(n, k)$-perfect hash family. An $(n, k)$-perfect hash family $\cF$ is a family of functions from $[n]$ to $k$ such that for every set $S \subseteq [n]$ of size $k$, there exists a function $f \in \cF$ that \textit{splits} $S$ \textit{evenly}. That is, for every $1 \leq j, j' \leq k$, $|f^{-1}(j) \cap S|$ and $|f^{-1}(j') \cap S|$ differ by at most $1$. For any $n, k \geq 1$, one can construct an $(n, k)$-perfect hash family of size $e^kk^{\mathcal{O}(log k)}log n$ in time $e^kk^{\mathcal{O}(log k)}n log n$~\cite{DBLP:books/sp/CyganFKLMPPS15,DBLP:conf/focs/NaorSS95}. In place of randomly coloring the voters and candidates with $k$ and $t$ colors respectively we construct $(m,k)$ and $(n,t)$-perfect hash families and run the algorithm exhaustively for all possible colorings generated by the functions in the hash families. By definition of perfect hash families it will generate a coloring where each candidate in  $O$ and each voter $v_i$ will receive distinct colors. If we check all $k!t!$ permutations of the colors we will get a good coloring. 
\end{proof}
Note that we can assume $t\leq k \Delta_C$ where $\Delta_C$ is the highest degree of a vertex in $\C$, otherwise,  it is a \no instance. Thus we get the following corollary. 
\begin{corollary}
    \probonegen admits an \fpt algorithm parameterized by $k+ \Delta_C$.
\end{corollary}

%% file: outlook.tex
\section{Outlook}

In this paper, we modeled the \textsc{Multiwinner Election} problem as a graph-theoretic problem which enables us to address the problem in the $K_{d,d}$-free graph class. 
This approach captures a broader range of profiles than that of \cite{DBLP:journals/iandc/Skowron17}. Specifically, it generalizes the class of bounded approval sets, a class that admits tractable results. For \probonegen, we developed an FPT-AS and a lossy polynomial-time preprocessing procedure. To the best of our knowledge, our additive approximation algorithm and lossy preprocessing method represent novel technical contributions to the field of computational social choice theory.

Our work is just a starting point in this area, with several potential extensions. In our algorithm, we assumed that the functions are both monotone and submodular. A natural question is what happens if we relax one of these constraints. Additionally, while we focused on the approval model of elections, the next logical step is to extend our investigations to ordinal or cardinal elections. Another direction would be to incorporate fairness or matroid constraints into the voting profiles, as explored in~\cite{DBLP:conf/icalp/00020LS0U24}. Also considering diversity constraints on selected committee as studied in \cite{DBLP:conf/aaai/BredereckFILS18} could be another direction of future work.

%% file: main.bbl
\begin{thebibliography}{27}
\providecommand{\natexlab}[1]{#1}
\providecommand{\url}[1]{\texttt{#1}}
\expandafter\ifx\csname urlstyle\endcsname\relax
  \providecommand{\doi}[1]{doi: #1}\else
  \providecommand{\doi}{doi: \begingroup \urlstyle{rm}\Url}\fi

\bibitem[Bal(2024)]{Ballotpedia}
Multi-winner system.
\newblock \url{https://ballotpedia.org/Multi-winner_system}, 2024.

\bibitem[Alon et~al.(1995)Alon, Yuster, and Zwick]{DBLP:journals/jacm/AlonYZ95}
Noga Alon, Raphael Yuster, and Uri Zwick.
\newblock Color-coding.
\newblock \emph{J. {ACM}}, 42\penalty0 (4):\penalty0 844--856, 1995.
\newblock \doi{10.1145/210332.210337}.
\newblock URL \url{https://doi.org/10.1145/210332.210337}.

\bibitem[Aziz et~al.(2015)Aziz, Gaspers, Gudmundsson, Mackenzie, Mattei, and
  Walsh]{DBLP:conf/atal/AzizGGMMW15}
Haris Aziz, Serge Gaspers, Joachim Gudmundsson, Simon Mackenzie, Nicholas
  Mattei, and Toby Walsh.
\newblock Computational aspects of multi-winner approval voting.
\newblock In Gerhard Weiss, Pinar Yolum, Rafael~H. Bordini, and Edith Elkind,
  editors, \emph{Proceedings of the 2015 International Conference on Autonomous
  Agents and Multiagent Systems, {AAMAS} 2015, Istanbul, Turkey, May 4-8,
  2015}, pages 107--115. {ACM}, 2015.
\newblock URL \url{http://dl.acm.org/citation.cfm?id=2772896}.

\bibitem[Betzler et~al.(2013)Betzler, Slinko, and
  Uhlmann]{DBLP:journals/jair/BetzlerSU13}
Nadja Betzler, Arkadii Slinko, and Johannes Uhlmann.
\newblock On the computation of fully proportional representation.
\newblock \emph{J. Artif. Intell. Res.}, 47:\penalty0 475--519, 2013.
\newblock \doi{10.1613/JAIR.3896}.
\newblock URL \url{https://doi.org/10.1613/jair.3896}.

\bibitem[Bredereck et~al.(2018)Bredereck, Faliszewski, Igarashi, Lackner, and
  Skowron]{DBLP:conf/aaai/BredereckFILS18}
Robert Bredereck, Piotr Faliszewski, Ayumi Igarashi, Martin Lackner, and Piotr
  Skowron.
\newblock Multiwinner elections with diversity constraints.
\newblock In Sheila~A. McIlraith and Kilian~Q. Weinberger, editors,
  \emph{Proceedings of the Thirty-Second {AAAI} Conference on Artificial
  Intelligence, (AAAI-18), the 30th innovative Applications of Artificial
  Intelligence (IAAI-18), and the 8th {AAAI} Symposium on Educational Advances
  in Artificial Intelligence (EAAI-18), New Orleans, Louisiana, USA, February
  2-7, 2018}, pages 933--940. {AAAI} Press, 2018.
\newblock \doi{10.1609/AAAI.V32I1.11457}.
\newblock URL \url{https://doi.org/10.1609/aaai.v32i1.11457}.

\bibitem[Chamberlin and
  Courant(1983)]{RePEc:cup:apsrev:v:77:y:1983:i:03:p:718-733_24}
John~R. Chamberlin and Paul~N. Courant.
\newblock {Representative Deliberations and Representative Decisions:
  Proportional Representation and the Borda Rule}.
\newblock \emph{American Political Science Review}, 77\penalty0 (3):\penalty0
  718--733, September 1983.
\newblock URL
  \url{https://ideas.repec.org/a/cup/apsrev/v77y1983i03p718-733_24.html}.

\bibitem[Cygan et~al.(2015)Cygan, Fomin, Kowalik, Lokshtanov, Marx, Pilipczuk,
  Pilipczuk, and Saurabh]{DBLP:books/sp/CyganFKLMPPS15}
Marek Cygan, Fedor~V. Fomin, Lukasz Kowalik, Daniel Lokshtanov, D{\'{a}}niel
  Marx, Marcin Pilipczuk, Michal Pilipczuk, and Saket Saurabh.
\newblock \emph{Parameterized Algorithms}.
\newblock Springer, 2015.
\newblock ISBN 978-3-319-21274-6.
\newblock \doi{10.1007/978-3-319-21275-3}.
\newblock URL \url{https://doi.org/10.1007/978-3-319-21275-3}.

\bibitem[Do et~al.(2022)Do, Hervouin, Lang, and
  Skowron]{DBLP:conf/ijcai/DoHL022}
Virginie Do, Matthieu Hervouin, J{\'{e}}r{\^{o}}me Lang, and Piotr Skowron.
\newblock Online approval committee elections.
\newblock In Luc~De Raedt, editor, \emph{Proceedings of the Thirty-First
  International Joint Conference on Artificial Intelligence, {IJCAI} 2022,
  Vienna, Austria, 23-29 July 2022}, pages 251--257. ijcai.org, 2022.
\newblock \doi{10.24963/IJCAI.2022/36}.
\newblock URL \url{https://doi.org/10.24963/ijcai.2022/36}.

\bibitem[Elkind et~al.(2017)Elkind, Faliszewski, Skowron, and
  Slinko]{DBLP:journals/scw/ElkindFSS17}
Edith Elkind, Piotr Faliszewski, Piotr Skowron, and Arkadii Slinko.
\newblock Properties of multiwinner voting rules.
\newblock \emph{Soc. Choice Welf.}, 48\penalty0 (3):\penalty0 599--632, 2017.
\newblock \doi{10.1007/S00355-017-1026-Z}.
\newblock URL \url{https://doi.org/10.1007/s00355-017-1026-z}.

\bibitem[Faliszewski et~al.(2017)Faliszewski, Skowron, Slinko, and
  Talmon]{16de4c2a459441008078dd35182e783f}
Piotr Faliszewski, Piotr Skowron, Arkadii Slinko, and Nimrod Talmon.
\newblock \emph{Multiwinner voting: A new challenge for social choice theory},
  volume~74 of \emph{Trends in computational social choice}, pages 27--47.
\newblock Lulu Publisher, 2017.
\newblock ISBN 1326912097.

\bibitem[Inamdar et~al.(2024)Inamdar, Jain, Lokshtanov, Sahu, Saurabh, and
  Upasana]{DBLP:conf/icalp/00020LS0U24}
Tanmay Inamdar, Pallavi Jain, Daniel Lokshtanov, Abhishek Sahu, Saket Saurabh,
  and Anannya Upasana.
\newblock Satisfiability to coverage in presence of fairness, matroid, and
  global constraints.
\newblock In Karl Bringmann, Martin Grohe, Gabriele Puppis, and Ola Svensson,
  editors, \emph{51st International Colloquium on Automata, Languages, and
  Programming, {ICALP} 2024, July 8-12, 2024, Tallinn, Estonia}, volume 297 of
  \emph{LIPIcs}, pages 88:1--88:18. Schloss Dagstuhl - Leibniz-Zentrum
  f{\"{u}}r Informatik, 2024.
\newblock \doi{10.4230/LIPICS.ICALP.2024.88}.
\newblock URL \url{https://doi.org/10.4230/LIPIcs.ICALP.2024.88}.

\bibitem[Jain et~al.(2023)Jain, Kanesh, Panolan, Saha, Sahu, Saurabh, and
  Upasana]{DBLP:conf/soda/0001KPSS0U23}
Pallavi Jain, Lawqueen Kanesh, Fahad Panolan, Souvik Saha, Abhishek Sahu, Saket
  Saurabh, and Anannya Upasana.
\newblock Parameterized approximation scheme for biclique-free max
  \emph{k}-weight {SAT} and max coverage.
\newblock In Nikhil Bansal and Viswanath Nagarajan, editors, \emph{Proceedings
  of the 2023 {ACM-SIAM} Symposium on Discrete Algorithms, {SODA} 2023,
  Florence, Italy, January 22-25, 2023}, pages 3713--3733. {SIAM}, 2023.
\newblock \doi{10.1137/1.9781611977554.CH143}.
\newblock URL \url{https://doi.org/10.1137/1.9781611977554.ch143}.

\bibitem[Janson(2016)]{janson2016phragmen}
Svante Janson.
\newblock Phragm{\'e}n’s and thiele’s election methods.
\newblock Technical report, Technical report, 2016.

\bibitem[Lackner and Skowron(2021)]{DBLP:journals/jet/LacknerS21}
Martin Lackner and Piotr Skowron.
\newblock Consistent approval-based multi-winner rules.
\newblock \emph{J. Econ. Theory}, 192:\penalty0 105173, 2021.
\newblock \doi{10.1016/J.JET.2020.105173}.
\newblock URL \url{https://doi.org/10.1016/j.jet.2020.105173}.

\bibitem[Lackner and Skowron(2023)]{DBLP:series/sbis/LacknerS23}
Martin Lackner and Piotr Skowron.
\newblock \emph{Multi-Winner Voting with Approval Preferences - Artificial
  Intelligence, Multiagent Systems, and Cognitive Robotics}.
\newblock Springer Briefs in Intelligent Systems. Springer, 2023.
\newblock ISBN 978-3-031-09015-8.
\newblock \doi{10.1007/978-3-031-09016-5}.
\newblock URL \url{https://doi.org/10.1007/978-3-031-09016-5}.

\bibitem[Lokshtanov et~al.(2017)Lokshtanov, Panolan, Ramanujan, and
  Saurabh]{DBLP:conf/stoc/LokshtanovPRS17}
Daniel Lokshtanov, Fahad Panolan, M.~S. Ramanujan, and Saket Saurabh.
\newblock Lossy kernelization.
\newblock In Hamed Hatami, Pierre McKenzie, and Valerie King, editors,
  \emph{Proceedings of the 49th Annual {ACM} {SIGACT} Symposium on Theory of
  Computing, {STOC} 2017, Montreal, QC, Canada, June 19-23, 2017}, pages
  224--237. {ACM}, 2017.
\newblock \doi{10.1145/3055399.3055456}.
\newblock URL \url{https://doi.org/10.1145/3055399.3055456}.

\bibitem[Manurangsi(2020)]{DBLP:conf/soda/Manurangsi20}
Pasin Manurangsi.
\newblock Tight running time lower bounds for strong inapproximability of
  maximum \emph{k}-coverage, unique set cover and related problems (via
  \emph{t}-wise agreement testing theorem).
\newblock In Shuchi Chawla, editor, \emph{Proceedings of the 2020 {ACM-SIAM}
  Symposium on Discrete Algorithms, {SODA} 2020, Salt Lake City, UT, USA,
  January 5-8, 2020}, pages 62--81. {SIAM}, 2020.
\newblock \doi{10.1137/1.9781611975994.5}.
\newblock URL \url{https://doi.org/10.1137/1.9781611975994.5}.

\bibitem[Manurangsi(2024)]{DBLP:journals/corr/abs-2403-06335}
Pasin Manurangsi.
\newblock Improved {FPT} approximation scheme and approximate kernel for
  biclique-free max k-weight {SAT:} greedy strikes back.
\newblock \emph{CoRR}, abs/2403.06335, 2024.
\newblock \doi{10.48550/ARXIV.2403.06335}.
\newblock URL \url{https://doi.org/10.48550/arXiv.2403.06335}.

\bibitem[Manurangsi(2025)]{DBLP:journals/tcs/Manurangsi25}
Pasin Manurangsi.
\newblock Improved {FPT} approximation scheme and approximate kernel for
  biclique-free max k-weight {SAT:} greedy strikes back.
\newblock \emph{Theor. Comput. Sci.}, 1028:\penalty0 115033, 2025.

\bibitem[Naor et~al.(1995)Naor, Schulman, and
  Srinivasan]{DBLP:conf/focs/NaorSS95}
Moni Naor, Leonard~J. Schulman, and Aravind Srinivasan.
\newblock Splitters and near-optimal derandomization.
\newblock In \emph{36th Annual Symposium on Foundations of Computer Science,
  Milwaukee, Wisconsin, USA, 23-25 October 1995}, pages 182--191. {IEEE}
  Computer Society, 1995.
\newblock \doi{10.1109/SFCS.1995.492475}.
\newblock URL \url{https://doi.org/10.1109/SFCS.1995.492475}.

\bibitem[Nemhauser et~al.(1978)Nemhauser, Wolsey, and
  Fisher]{DBLP:journals/mp/NemhauserWF78}
George~L. Nemhauser, Laurence~A. Wolsey, and Marshall~L. Fisher.
\newblock An analysis of approximations for maximizing submodular set functions
  - {I}.
\newblock \emph{Math. Program.}, 14\penalty0 (1):\penalty0 265--294, 1978.
\newblock \doi{10.1007/BF01588971}.
\newblock URL \url{https://doi.org/10.1007/BF01588971}.

\bibitem[Pierczynski and Skowron(2019)]{DBLP:conf/ijcai/PierczynskiS19}
Grzegorz Pierczynski and Piotr Skowron.
\newblock Approval-based elections and distortion of voting rules.
\newblock In Sarit Kraus, editor, \emph{Proceedings of the Twenty-Eighth
  International Joint Conference on Artificial Intelligence, {IJCAI} 2019,
  Macao, China, August 10-16, 2019}, pages 543--549. ijcai.org, 2019.
\newblock \doi{10.24963/IJCAI.2019/77}.
\newblock URL \url{https://doi.org/10.24963/ijcai.2019/77}.

\bibitem[Procaccia et~al.(2007)Procaccia, Rosenschein, and
  Zohar]{DBLP:conf/ijcai/ProcacciaRZ07}
Ariel~D. Procaccia, Jeffrey~S. Rosenschein, and Aviv Zohar.
\newblock Multi-winner elections: Complexity of manipulation, control and
  winner-determination.
\newblock In Manuela~M. Veloso, editor, \emph{{IJCAI} 2007, Proceedings of the
  20th International Joint Conference on Artificial Intelligence, Hyderabad,
  India, January 6-12, 2007}, pages 1476--1481, 2007.
\newblock URL \url{http://ijcai.org/Proceedings/07/Papers/238.pdf}.

\bibitem[Skowron(2017)]{DBLP:journals/iandc/Skowron17}
Piotr Skowron.
\newblock {FPT} approximation schemes for maximizing submodular functions.
\newblock \emph{Inf. Comput.}, 257:\penalty0 65--78, 2017.
\newblock \doi{10.1016/J.IC.2017.10.002}.
\newblock URL \url{https://doi.org/10.1016/j.ic.2017.10.002}.

\bibitem[Skowron and Faliszewski(2017)]{DBLP:journals/jair/SkowronF17}
Piotr Skowron and Piotr Faliszewski.
\newblock Chamberlin-courant rule with approval ballots: Approximating the
  maxcover problem with bounded frequencies in {FPT} time.
\newblock \emph{J. Artif. Intell. Res.}, 60:\penalty0 687--716, 2017.
\newblock \doi{10.1613/JAIR.5628}.
\newblock URL \url{https://doi.org/10.1613/jair.5628}.

\bibitem[Skowron and Faliszewski(2015)]{DBLP:conf/aaai/SkowronF15}
Piotr~Krzysztof Skowron and Piotr Faliszewski.
\newblock Fully proportional representation with approval ballots:
  Approximating the maxcover problem with bounded frequencies in {FPT} time.
\newblock In Blai Bonet and Sven Koenig, editors, \emph{Proceedings of the
  Twenty-Ninth {AAAI} Conference on Artificial Intelligence, January 25-30,
  2015, Austin, Texas, {USA}}, pages 2124--2130. {AAAI} Press, 2015.
\newblock \doi{10.1609/AAAI.V29I1.9432}.
\newblock URL \url{https://doi.org/10.1609/aaai.v29i1.9432}.

\bibitem[Yang and Wang(2023)]{DBLP:journals/aamas/YangW23}
Yongjie Yang and Jian{-}xin Wang.
\newblock Parameterized complexity of multiwinner determination: more effort
  towards fixed-parameter tractability.
\newblock \emph{Auton. Agents Multi Agent Syst.}, 37\penalty0 (2):\penalty0 28,
  2023.
\newblock \doi{10.1007/S10458-023-09610-Z}.
\newblock URL \url{https://doi.org/10.1007/s10458-023-09610-z}.

\end{thebibliography}
